\documentclass[sigconf]{acmart}
\AtBeginDocument{%
  \providecommand\BibTeX{{%
    \normalfont B\kern-0.5em{\scshape i\kern-0.25em b}\kern-0.8em\TeX}}}

\setcopyright{acmcopyright}
\copyrightyear{2021}
\acmYear{2021}
\acmDOI{xxxx}

\copyrightyear{2021}
\acmYear{2021}
\setcopyright{acmcopyright}
\acmConference[CIKM '21] {Proceedings of the 30th ACM International Conference on Information and Knowledge Management}{November 1--5, 2021}{Virtual Event, Australia.}
\acmBooktitle{Proceedings of the 30th ACM Int'l Conf. on Information and Knowledge Management (CIKM '21), November 1--5, 2021, Virtual Event, Australia}
\acmPrice{15.00}
\acmISBN{978-1-4503-8446-9/21/11}
\acmDOI{10.1145/3459637.3482288}

\usepackage[linesnumbered,ruled,vlined]{algorithm2e}
\usepackage{tabularx}
\usepackage{subfigure}
\usepackage{enumitem}
\usepackage{url}
\usepackage{balance}

\newcommand{\figureTopMargin}{\vspace{-3ex}}
\newcommand{\figureCaptionMargin}{\vspace{-3ex}}
\newcommand{\figureBelowMargin}{\vspace{-2ex}}

\newcommand{\secref}[1]{Sec.~\ref{#1}}

\newcommand{\figref}[1]{Fig.~\ref{#1}}

\newcommand{\fakeparagraph}[1]{\vspace{0.3mm}\noindent\textbf{#1.}}

\newcommand{\eg}{\emph{e.g.},\xspace}

\begin{document}
\fancyhead{}

\title{Privacy-Preserving Batch-based Task Assignment in Spatial Crowdsourcing with Untrusted Server}


\author{Maocheng Li}
\email{csmichael@cse.ust.hk}
\affiliation{%
    \institution{HKUST}
    \city{Hong Kong}
    \country{China}
}

\author{Jiachuan Wang}
\email{jwangey@cse.ust.hk}
\affiliation{%
    \institution{HKUST}
    \city{Hong Kong}
    \country{China}
}

\author{Libin Zheng}
\email{zhenglb6@mail.sysu.edu.cn}
\affiliation{%
    \institution{Sun Yat-sen University}
    \city{Guangzhou}
    \country{China}
}

\author{Han Wu}
\email{han.wu@stu.ecnu.edu.cn}
\affiliation{%
    \institution{East China Normal University}
    \city{Shanghai}
    \country{China}
}

\author{Peng Cheng}
\email{pcheng@sei.ecnu.edu.cn}
\affiliation{%
    \institution{East China Normal University}
    \city{Shanghai}
    \country{China}
}

\author{Lei Chen}
\orcid{0001-6984-7670}
\email{leichen@cse.ust.hk}
\affiliation{%
    \institution{HKUST}
    \streetaddress{Clear Water Bay}
    \city{Hong Kong}
    \country{China}
}

\author{Xuemin Lin}
\email{lxue@cse.unsw.edu.au}
\affiliation{%
    \institution{The University of New South Wales}
    \city{Sydney}
    \country{Australia}
}

\renewcommand{\shortauthors}{Maocheng Li et al.}

\begin{abstract}
In this paper, we study the privacy-preserving task assignment problem in spatial crowdsourcing, where the locations of both workers and tasks, prior to their release to the server, are perturbed with Geo-Indistinguishability (a differential privacy notion for location-based systems). Different from the previously studied online setting, where each task is assigned immediately upon arrival, we target the batch-based setting, where the server maximizes the number of successfully assigned tasks after a \textit{batch} of tasks arrive. To achieve this goal, we propose the $k$-Switch solution, which first divides the workers into small groups based on the perturbed distance between workers/tasks, and then  utilizes Homomorphic Encryption (HE) based secure computation to enhance the task assignment. Furthermore, we expedite HE-based computation by limiting the size of the small groups under $k$. Extensive experiments demonstrate that, in terms of the number of successfully assigned tasks, the $k$-Switch solution improves batch-based baselines by 5.9$\times$ and the existing online solution by 1.74$\times$, with no privacy leak.
\end{abstract}


\maketitle


\section{Introduction}
\label{sec:intro}
The mass adoption of GPS-equipped smart phones enables individuals to collaborate, participate, consume and produce valuable information about the environment and themselves. \textit{Spatial crowdsourcing} (SC) systems (e.g., Foursquare \cite{foursquare} , Gigwalk \cite{gigwalk}, MediaQ \cite{kim2014mediaq}, and gMission \cite{chen2014gmission}) have emerged to support such collaborations by assigning tasks to proper workers.  Task assignment is a  core issue in SC systems, asking the workers to 
move physically to specified locations to execute the tasks \cite{kazemi2012, DBLP:journals/vldb/yongxin19}.

To enable the SC server to properly assign tasks, in general, workers need to upload their locations. However, users' location is highly sensitive, because it can indicate users' whereabouts, and even disclose their private attributes. For example,  visiting an urgent care center reveals certain medical conditions \cite{DBLP:journals/fttcs/DworkR14}. The server, which receives the locations,  is untrusted and can be vulnerable to attacks. Thus, in previous studies \cite{to2018, DBLP:conf/icde/TaoTZSC020}, privacy-preserving task assignment is proposed to enable users (workers/task requesters) to perturb their locations with \textit{Geo-Indistinguishability} (Geo-I) \cite{andres13} and upload only the perturbed locations. Geo-I is a widely adopted differential privacy notion \cite{DBLP:conf/icalp/Dwork06} for location-based systems, which defends users' locations against strong adversaries with any prior knowledge \cite{DBLP:journals/fttcs/DworkR14}.

In this paper, we study the task assignment problem with users' perturbed locations under the batch-based setting. That is, the SC server assigns the tasks batch-by-batch to the workers, with locations of both workers and tasks  perturbed by Geo-I. Although Geo-I is an existing technique to protect location privacy of users, directly applying Geo-I to task assignment problem could lead to poor performance as measured by the number of successfully assigned tasks, because only the perturbed locations are available as the problem inputs. To the best of our knowledge,  existing works on spatial crowdsourcing with Geo-I  \cite{to2018,DBLP:conf/icde/TaoTZSC020} target the online setting where tasks come on the fly and are processed one by one. Batch-based task assignment with Geo-I-perturbed locations has not yet been studied. Directly applying the online assignment methods to the batch-based setting can result in poor solutions, because online solution assigns tasks one by one, and previous assignments can not be changed when new tasks arrive. An example of the online solution making sub-optimal assignments is shown in Figure~\ref{fig:online_offline}.

\begin{figure}[ht!]\centering \figureTopMargin\vspace{0ex}
	\subfigure[][{\scriptsize Online solution}]{
		\scalebox{0.4}[0.4]{\includegraphics{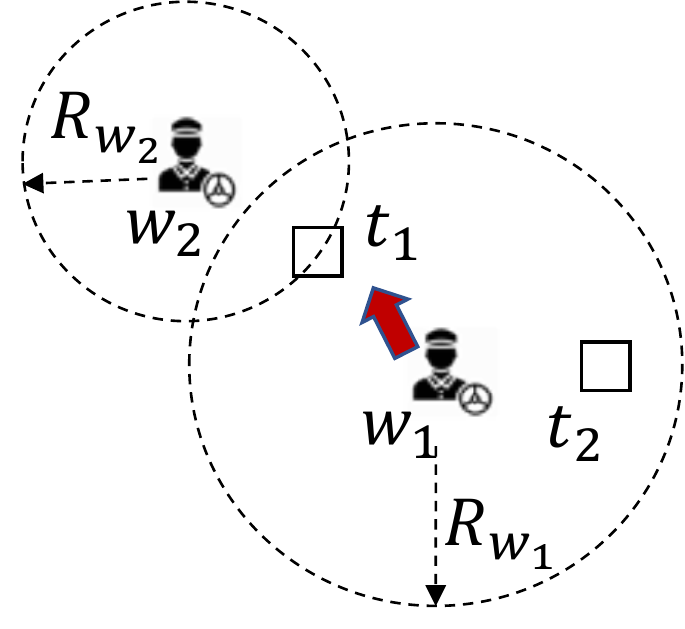}}
		\label{subfig:fig1_online_offline_left}}
	\subfigure[][{\scriptsize  Optimal solution}]{
		\scalebox{0.4}[0.4]{\includegraphics{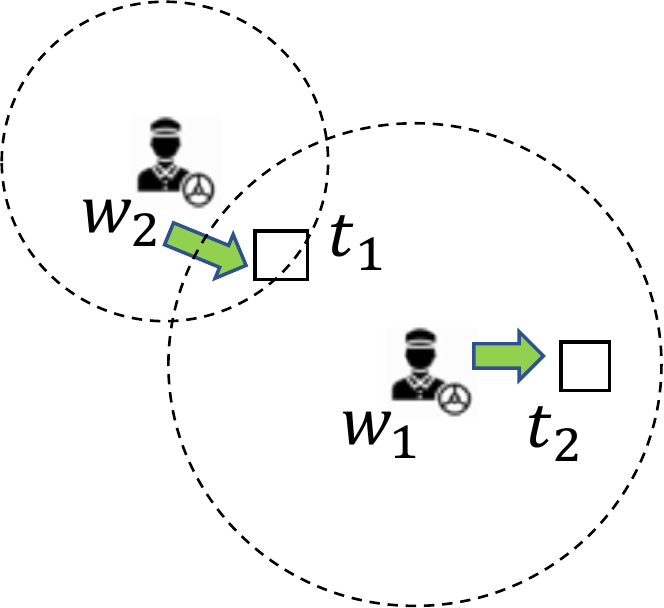}}
		\label{subfig:fig1_online_offline_right}}\figureCaptionMargin
    \caption{\small Online solution makes sub-optimal assignment vs. Optimal assignment. }
    \figureBelowMargin\vspace{0ex}
	\label{fig:online_offline}
\end{figure}

As Fig.~\ref{subfig:fig1_online_offline_left} shows, in the online setting, task $t_1$ arrives first and would be assigned to worker $w_1$ since they are close to each other. Note that such
a decision is irrevocable in the online setting. Subsequently, $w_1$ will no longer be available no matter what future tasks arrive. When task $t_2$ arrives later, because worker $w_1$ is occupied, $t_2$ cannot be assigned. $w_1$ is the only worker who can perform $t_2$, as $t_2$ is outside the reachability range $R_{w_2}$ of worker $w_2$. In contrast, when $t_1$ and $t_2$ come in a batch, they can be properly assigned to $w_2$ and $w_1$, respectively.


No existing research has specifically targeted the batch-based privacy-preserving task assignment, however, as a matter of fact, batch-based setting has been widely adopted in the spatial crowdsourcing industry. For example, DiDi \cite{didi}, the leading car-hailing platform in China, accumulates orders (tasks) per time window, and jointly dispatches them to drivers (workers) \cite{DBLP:journals/pvldb/Zheng0Y18, didi2018}. 
Figure~\ref{fig:offline} gives an example of our studied problem in this paper. Fig.~\ref{subfig:fig2_offline_left} shows the true locations of workers $w_1, w_2$ and tasks $t_1, t_2$.   $w_1$ and $w_2$  have a reachability range $R_{w_1}$ and $R_{w_2}$ respectively, denoting the maximum distance that they are willing to travel. In Fig.~\ref{subfig:fig2_offline_right}, the locations are perturbed using Geo-I. Our problem is to assign tasks to suitable workers to maximize the number of assigned tasks subjecting to  workers' reachability constraint.


\begin{figure}[ht!]\centering \figureTopMargin\vspace{0ex}
	\subfigure[][{\scriptsize True locations}]{
		\scalebox{0.4}[0.4]{\includegraphics{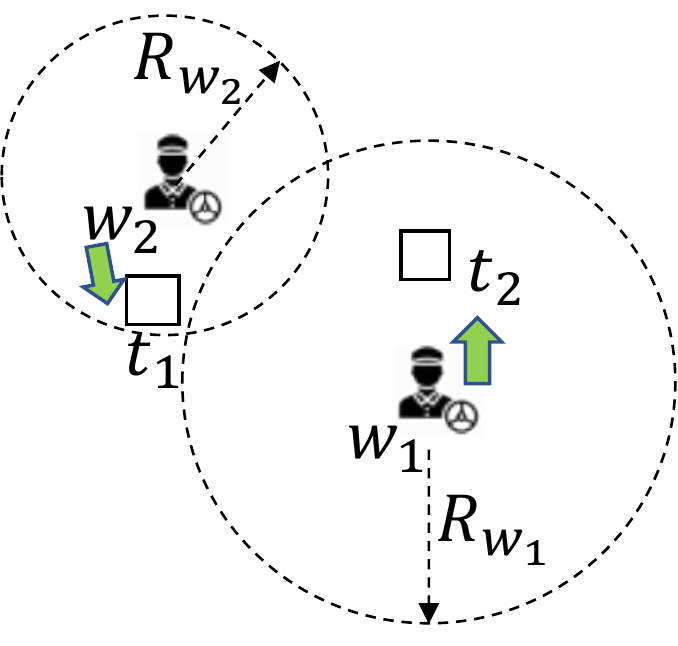}}
		\label{subfig:fig2_offline_left}}
	\subfigure[][{\scriptsize  Perturbed locations}]{
		\scalebox{0.4}[0.4]{\includegraphics{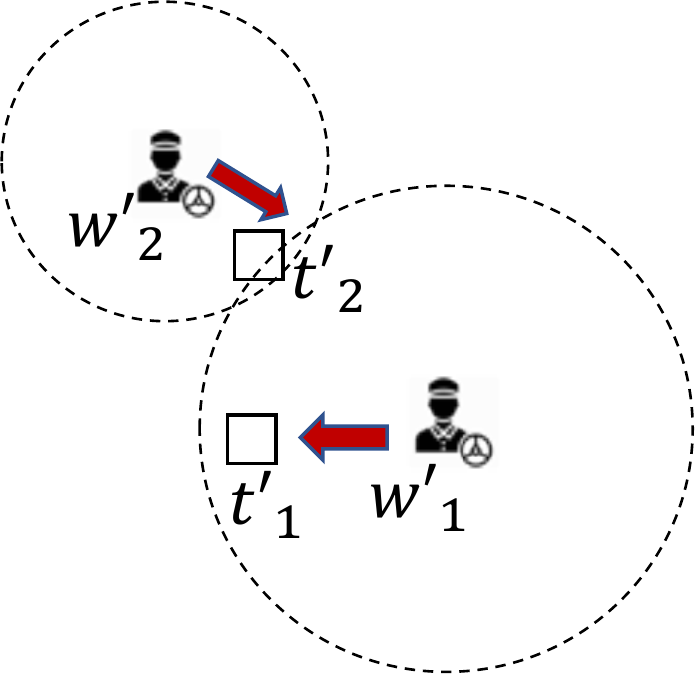}}
		\label{subfig:fig2_offline_right}}\figureCaptionMargin
    \caption{\small Our Privacy-preserving Batch-based Task Assignment (PBTA) problem. }
    \figureBelowMargin\vspace{0ex}
	\label{fig:offline}
\end{figure}

The obfuscated locations bring technical challenges to the task assignment problem. As shown in Fig.~\ref{subfig:fig2_offline_right}, if we directly use the perturbed locations as the true locations to perform the task assignment, then task $t_2$ and $t_1$ are assigned to worker $w_2$ and $w_1$ respectively (denoted by the red arrows). However, according to the true locations, $t_2$ lies outside the range of $w_2$, and $t_1$ lies outside the range of $w_1$. The aforementioned assignment finally has no task successfully executed.
In fact, the optimal assignment should be assigning $t_2$ to $w_1$ and $t_1$ to $w_2$ (denoted by the green arrows in Fig.~\ref{subfig:fig2_offline_left}), which would finally execute both tasks.

To address the aforementioned challenge, we propose an innovative solution $k$-Switch. It first employs probabilistic inference over the perturbed locations to obtain an initial assignment, then uses an encryption-based secure computation method to improve the assignment quality significantly without compromising privacy. To address the efficiency issue of existing Homomorphic Encryption techniques, it divides workers into small groups of size $k$, and apply HE within the small groups only. Using the example in Fig.~\ref{fig:offline}, $k$-Switch uses perturbed locations to obtain a preliminary assignment first, and then allows worker $w_1$ and $w_2$ to use secure communication to obtain the \textit{true distance} from the tasks and  switch tasks in order to obtain a better assignment. 

To summarize, we make the following contributions in this paper.

\begin{itemize}[leftmargin=*]
\setlength\itemsep{0.5em}
\item We propose $k$-Switch, an innovative solution to address the challenging Privacy-preserving Batch-based Task Assignment (PBTA) problem. To the best of our knowledge, this is the first work targeting the \textit{batch-based setting} of the privacy-preserving task assignment problem. The PBTA problem is defined in  Section~\ref{sec:problem} together with two proposed oblivious baselines.  

\item  On top of the oblivious baselines, we propose the  $k$-Switch method, which combines probabilistic analysis and encryption-based secure computation techniques to achieve significant assignment quality improvements without compromising privacy. It mitigates the efficiency issue of previous encryption-based secure computation techniques by doing the computation within the small groups. We introduce $k$-Switch in Section~\ref{sec:kSwitch}.

\item We conduct extensive experiments to validate the efficiency and effectiveness of the proposed algorithms. In terms of the number of successfully assigned tasks, the \textit{k}-switch method assigns up to 5.9$\times$ more tasks than other proposed baselines, and assigns up to 1.74$\times$ more tasks than the existing online solution. The results are shown in Section \ref{sec:experiment}.
\end{itemize}

In addition, we cover the preliminaries and necessary background in Sec.~\ref{sec:background}, compare with other related works in Sec.~\ref{sec:relatedWork} and conclude the paper in Sec.~\ref{sec:conclusion}. The notations used in this paper are summarized in Table \ref{tab:table1}.
\section{Background} 
\label{sec:background}

\begin{table}
	\centering \vspace{-3ex}
	{\small\scriptsize
		\caption{\small Notation.} \label{tab:table1}
		\vspace{-3ex}
		\begin{tabular}{l|m{6cm}}
			\hline
			{\bf Symbol} & {\bf \qquad \qquad \qquad\qquad\qquad Description} \\ \hline 
			$w, t, W, T$   & A worker, a task, a set of workers, a set of tasks\\
			\hline
			$w', t'$   & A worker and a task after perturbation \\
			\hline
			$l_{w}, l_{t}$  & True locations for a worker $w$ and a task $t$\\
			\hline
			$l_{w'}, l_{t'}$   & Perturbed locations for a worker $w$ and a task $t$\\
			\hline
			$R_{w}$  & Reachable distance for a worker $w$\\
			\hline
			$d(\cdot)$  & Euclidean distance function\\
			\hline
			$M$  & An assignment (matching) between workers and tasks \\
			\hline
			$M_0$  & A baseline matching obtained by oblivious baseline methods \\
			\hline
			$p, P$ & $p$ is a matched worker/task pair $p=(w, t)$. $P = W \times T$ is the set of all possible worker-task pairs. \\
			\hline
			$l=\epsilon r$  & Privacy parameters \\
			\hline
		\end{tabular}
	}\vspace{-1ex}
\end{table}

\subsection{Task assignments in spatial crowdsourcing}
We introduce the task assignment problem, especially the batch-based setting, which is the primary focus of this paper. 

\begin{definition} (Task assignment problem in spatial crowdsourcing \cite{DBLP:journals/vldb/yongxin19, DBLP:conf/icde/ZengTCZ18}) Given a set of workers $W$ and a set of tasks $T$, the task assignment problem returns an assignment (matching) $M$ of tasks to workers $M=\{(w, t)|w \in W, t \in T\}$ such that for some given objective function $\Psi(\cdot)$, $\Psi(M)$ is optimized (maximized or minimized):
	$$\Psi(M) = \sum_{(w, t)\in M}\psi(w,t),$$

\end{definition}

Following the setting in \cite{to2018}, we have the following spatial constraint and objective function. Each worker is willing to travel at most $R_w \in \mathbb{R}$, i.e., for each pair $(w, t) \in M$, $d(l_w, l_t) \leq R_w$. Each successfully assigned task carries a unit utility. The task assignment problem maximizes $\Psi(M) = |M|$, the size of the matching. 

In this paper, we focus on the batch-based setting (also referred as offline or static setting): the locations of \textit{all} workers and tasks are known at the beginning. Previous works \cite{to2018,DBLP:conf/icde/TaoTZSC020} mainly focus on the privacy preservation for the online setting: tasks arrive one by one, and each task needs to be assigned immediately upon arrival and cannot be re-assigned to other workers no matter what future tasks arrive.

\subsection{Geo-indistinguishability}
\label{subsec:geo-i}

Geo-indistinguishability (Geo-I) \cite{andres13} extends the traditional and well-adopted privacy notion -- differential privacy \cite{DBLP:conf/icalp/Dwork06} to location-based systems. 

\begin{definition} (Geo-I) For all true locations $x,x'$, a privacy parameter $\epsilon$, a mechanism $M$ satisfies  $\epsilon$-Geo-I iff:
	$$d_{\rho}(M(x),M(x'))\leq \epsilon d(x,x'),$$
\end{definition}
where $d(x,x')$ is the Euclidean distance between $x$ and $x'$ while $d_{\rho}(M(x),M(x'))$ is the multiplicative distance between two distributions $M(x)$ and $M(x')$. $M(x)$ and $M(x')$ are the distributions of perturbed locations based on the original location $x$ and $x'$ respectively.

One particular mechanism satisfying Geo-I is drawing random noise from the planar Laplace distribution \cite{andres13}. Given the privacy parameter $\epsilon \in \mathbb{R}^{+}$, the actual location $x_0 \in \mathbb{R}^{2}$, the probability density function of a noisy location $x \in \mathbb{R}^{2}$ is:

\begin{equation}
D_\epsilon(x_0)(x)=\frac{\epsilon ^2}{2\pi}e^{-\epsilon d(x_0,x)} \\, \label{eq:geo_i_equation}
\end{equation}

where $\frac{\epsilon ^2}{2\pi}$ is the normalization factor.

\section{Problem Definition}
\label{sec:problem}
We formally define the Privacy-preserving Batch-based Task Assignment (PBTA) problem in this section and introduce our proposed oblivious baselines. 

\subsection{PBTA problem}

\begin{definition} (Privacy-preserving Batch-based Task Assignment (PBTA) problem) Given a set of workers $W$ and a set of tasks $T$, the perturbed location $l_{t'}$ for each task $t$, the perturbed location $l_{w'}$ for each task $w$, the reachable distance $R_w$ for each worker $w$, the Euclidean distance function $d(\cdot)$, the PBTA problem is to return an assignment (matching) $M$ of tasks to workers $M=\{(w, t)|w \in W, t \in T, d(l_w, l_t) \leq R_w\}$ such that the following objective function is maximized:
	$$\Psi(M) = |M|.$$
\vspace{-5ex}
\label{def:pbta}
\end{definition}

Note that in the definition above, only the perturbed locations $l_{w'}, l_{t'}$ are available in the input. On the other hand, for the objective we are maximizing, each pair in the matching $M$ needs to satisfy the spatial constraint w.r.t. the true locations $l_w, l_t$. We have shown a simple example in Sec.~\ref{sec:intro} (Figure~\ref{fig:offline}) to illustrate the problem. 

In our setting, The perturbed locations of both workers $l_{w'}$ and tasks $l_{t'}$ are obtained by applying Geo-I with the privacy level $l = \epsilon r$. The privacy level is the same for all workers and tasks. Also, one worker $w$ takes at most one task, and each task $t$ only needs to be matched to one worker. The setting follows \cite{to2018}.

\subsection{Privacy model}

\subsubsection{System model}

We follow the system model assumptions of \cite{to2018}. We have three parties in our system: the server, the workers, and the task requesters (short as \textit{tasks} hereafter). 

Because the server is untrusted, when workers and tasks submit their locations, they only send the perturbed locations. The server is untrusted, as centralized servers are usually vulnerable to attacks, suffering massive data leak. For example, Facebook security breach exposes 50 million users' data \cite{facebook2018}. 

The task assignment is done at the server. After workers and tasks submit perturbed locations to the server, the server runs its task assignment algorithm, obtains an assignment $M$, and notifies workers and tasks about their assignments. 

\subsubsection{Adversary model}

Similar to the setting in \cite{to2018}, we adopt a semi-honest model, which assumes that all participating parties (the workers, the task requesters, and the server) are curious but not malicious.

They are curious about the private and sensitive information about other parties. So, we try to prevent sensitive information of any party from being shared with other parties. On the other hand, they are not malicious and follow system protocols. They do not collude with each other to gain extra information. 
\subsection{Baselines}

We propose a simple baseline solution Oblivious-M, which directly uses the perturbed (observed) locations to obtain a matching between workers and tasks. Then, we incorporate probabilistic analysis of reachability between workers and tasks, and propose another solution Oblivious-RR, which is based on randomized rounding. 

The term `Oblivious' indicates that both of these methods only access the perturbed locations, and true locations have never been accessed and disclosed in any way when running these methods. 

\subsubsection{Oblivious-M}
\label{subsec:om}

The Oblivious-M method is the simplest baseline solution we propose. It builds the reachability graph from the perturbed locations, and runs the Max-Flow algorithm (e.g., Ford-Fulkerson algorithm \cite{cormen01introduction}) to obtain a maximized cardinality matching. 

We use a running example in Figure~\ref{fig:example} (modified from Fig. 3 of \protect\cite{to2018}) to demonstrate the basic steps of Oblivious-M. In the PBTA problem, we are given the observed (perturbed) locations of workers and tasks (Fig.~\ref{subfig:example_left}). In Oblivious-M method, directly using the observed locations, if a task $t$ is within the range $R_w$ of worker $w$, we create an edge between $w$ and $t$ indicating they are reachable. In this way, we build the reachability graph as shown in Fig.~\ref{subfig:example_right}. 

\begin{figure}[t!]\centering \figureTopMargin\vspace{0ex}
	\subfigure[][{\scriptsize Inputs}]{
		\scalebox{0.35}[0.35]{\includegraphics{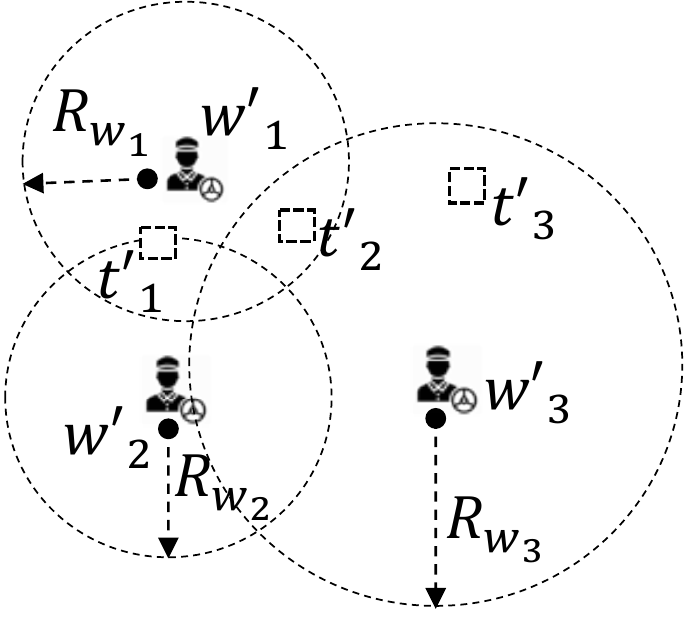}}
		\label{subfig:example_left}}
	\subfigure[][{\scriptsize  Reachability graph}]{
		\scalebox{0.4}[0.4]{\includegraphics{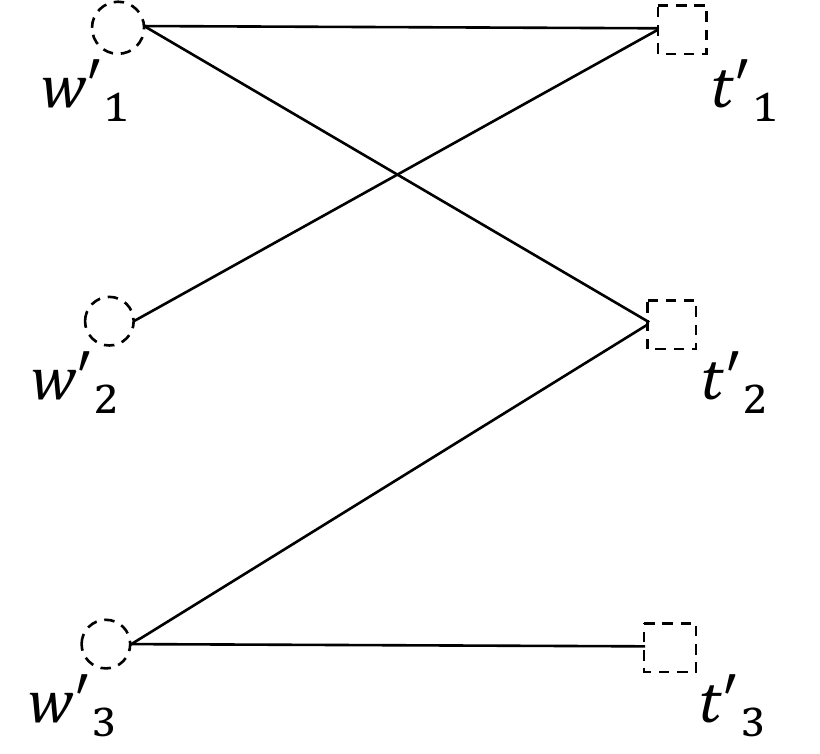}}
		\label{subfig:example_right}}\figureCaptionMargin
    \caption{\small A running example of PBTA problem}
    \figureBelowMargin\vspace{0ex}
	\label{fig:example}
\end{figure}

The reachability graph is a bipartite graph with all workers on one side and tasks on the other side, and all edges (indicating reachability) in between. Next, we simply add superficial source/sink nodes, edges with capacity 1 between the source node and every worker, and edges with capacity 1 between the sink node and every task to build the flow network (see Fig.~\ref{fig:example_3}). It is shown that the maximum cardinality matching from the reachability graph corresponds to the Max-Flow on the constructed flow network. So, we run standard Max-Flow algorithms (such as Ford-Fulkerson \cite{cormen01introduction}) on the constructed flow network, and return all saturated edges $(w, t)$ (with flow value 1 on the edges) as the output $M$. In our running example, the maximum flow obtained is shown in Fig.~\ref{fig:example_3}. The returned matching is $M=\{(w_1, t_2), (w_2, t_1), (w_3, t_3)\}$. 

\begin{figure}[ht!]\centering \figureTopMargin\vspace{0ex}
	\scalebox{0.35}[0.35]{\includegraphics{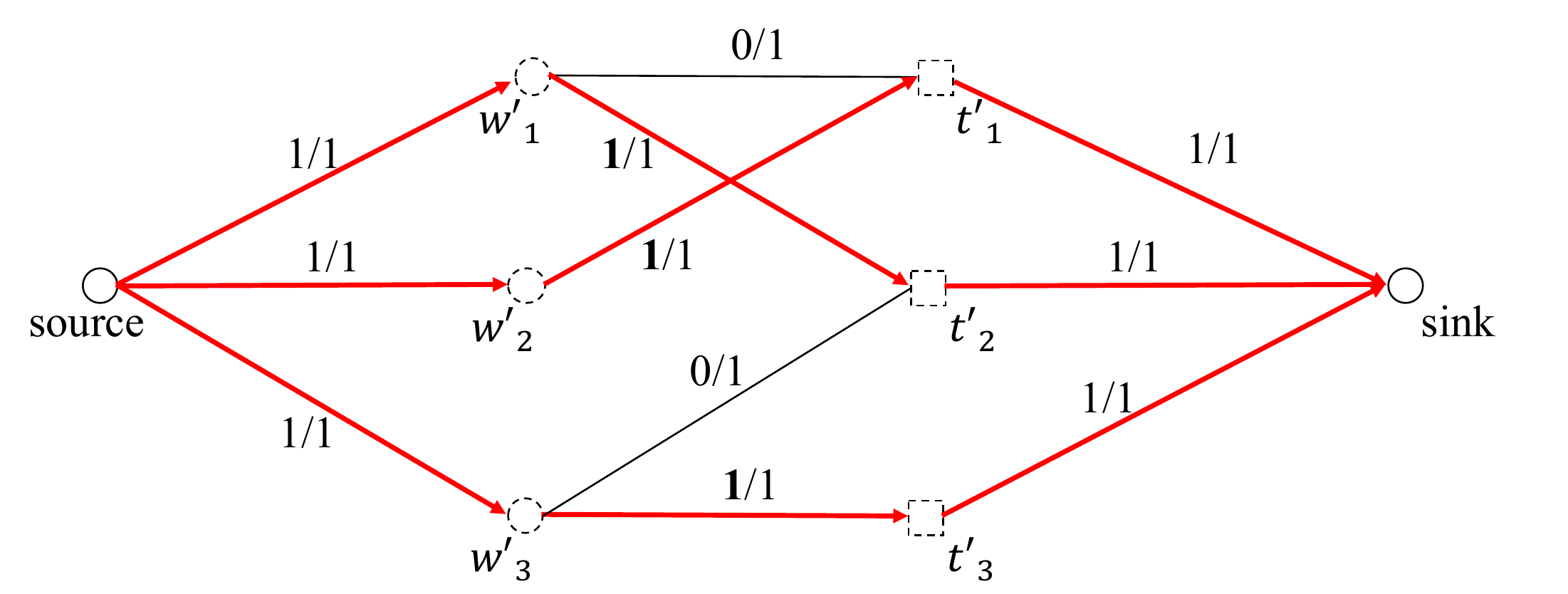}}
	\figureCaptionMargin\vspace{0ex}
	\caption{\small The maximum flow obtained on our running example.}
	\figureBelowMargin\vspace{0ex}
	\label{fig:example_3}
\end{figure}

When we measure the size of the matching $M$, we need to use the true locations to check whether a task is indeed within the reach of the assigned worker. The true locations and the reachability graphs built from the true locations are shown in Figure~\ref{fig:example_4}. In our returned matching $M$ based only on the perturbed locations, $t_2$ and $t_1$ are indeed within the range of $w_1$ and $w_2$ respectively. However, $t_3$ lies outside the range of $w_3$. Thus, the size of the matching we find is 2 (instead of 3). 

\begin{figure}[ht!]\centering \figureTopMargin\vspace{0ex}
	\scalebox{0.35}[0.35]{\includegraphics{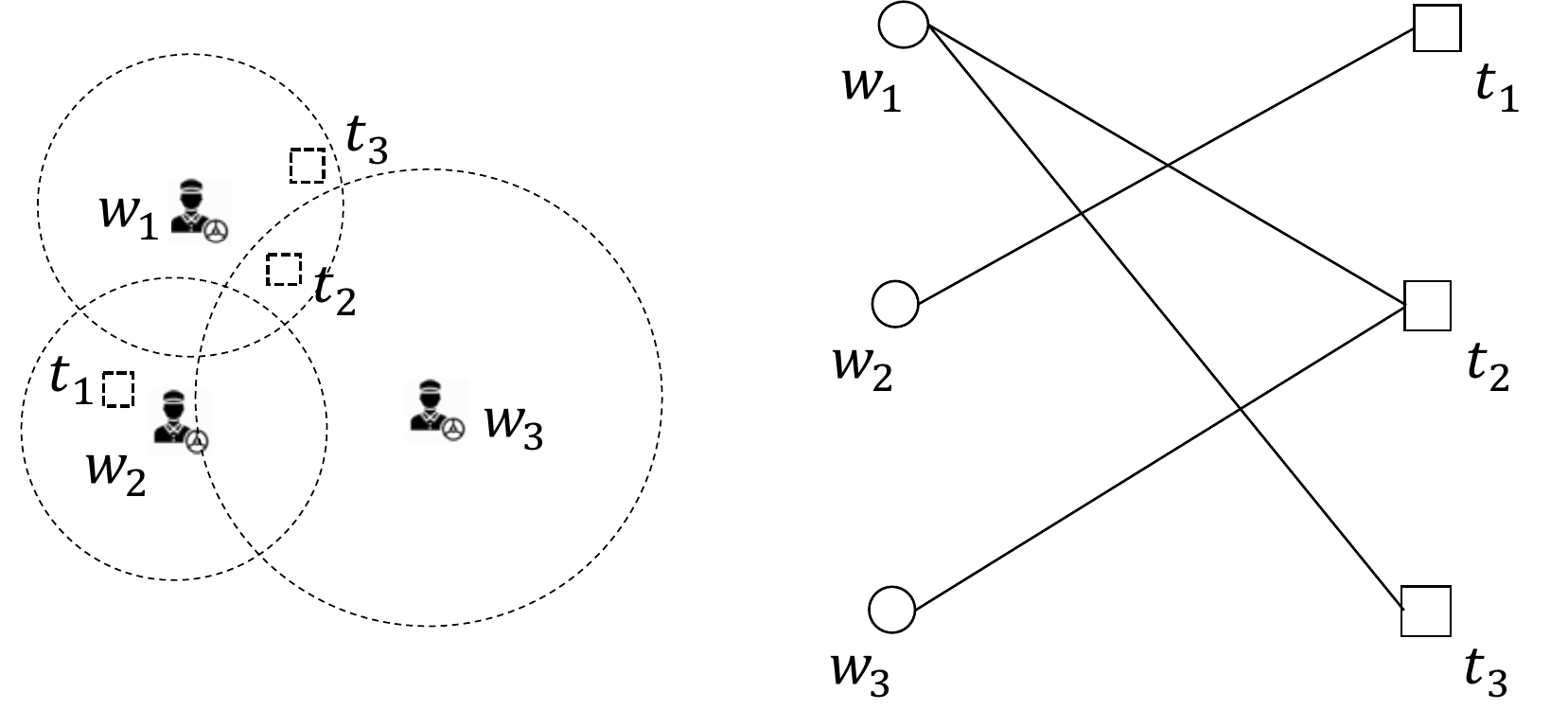}}
	\figureCaptionMargin
	\caption{\small The true locations of our running example.}
	\figureBelowMargin\vspace{0ex}
	\label{fig:example_4}
\end{figure}

For the time complexity, the fattest-path first implementation of Ford-Fulkerson is $\mathcal{O}(|E| \cdot opt)$, where $E$ denotes edges for the flow network, and $opt$ is the optimal value of the flow. We let $n = \max(|W|, |T|)$, the larger value of the sizes of the workers and the tasks. In our constructed flow network, assuming all tasks are connected with all workers in the worst case, $|E|=\mathcal{O}(n^2)$. As for $opt$, the optimal flow value is bounded by $n$. Overall, Oblivious-M has complexity of $\mathcal{O}(|E| \cdot opt) \to \mathcal{O}(n^2 \cdot n) \to \mathcal{O}(n^3)$.  

\subsubsection{Oblivious-RR}
\label{subsec:orr}

When we construct the reachability graph in the simplest baseline solution Oblivious-M, an edge between a worker $w$ and a task $t$ is either 1 (if perturbed locations show that $t$ is reachable from $w$) or 0 (when $t$ is not reachable from $w$). We make improvements in this step by adopting probabilistic analysis. Intuitively, based on perturbed locations, some worker-task pairs should have higher likelihood that they are indeed reachable, and on the other hand, some other worker-task pairs should have lower likelihood of being reachable.

We continue with the running example in Figure~\ref{fig:example} to further explain the intuition of the probabilistic analysis. If we focus on $w'_1$, and compare two tasks $t'_1$ and $t'_2$, then based on the perturbed locations, $t'_1$ looks much closer to $w'_1$ compared to $t'_2$ (the distance between $w'_1$ to $t'_1$ is almost 1/2 of $t'_2$, as $t'_2$ is almost at the periphery of the reachable circle). Oblivious-RR incorporates the probabilistic analysis and gives each edge in the reachability graph a fractional weight ranging from 0 to 1, instead of a binary 0/1 value. Based the fractional flow network, we design randomized techniques to obtain a matching. Due to space limit, we leave the details of Oblivious-RR in the appendix (\secref{subsec:appendix_orr}). Here, we show a particular matching obtained after the randomized rounding in Fig.~\ref{subfig:orr_2_right}, which assigns the same amount of tasks as the optimal matching. 

The time complexity of Oblivious-RR is the same with Oblivious-M, as it only adds a post randomized rounding which takes $\mathcal{O}(n^2)$ time, with $n = \max(|W|, |T|)$. Overall, Oblivious-RR runs in $\mathcal{O}(n^3)$ time. 

\begin{figure}[t!]\centering \figureTopMargin\vspace{0ex}
	\subfigure[][{\scriptsize The maximum weight flow}]{
		\scalebox{0.4}[0.4]{\includegraphics{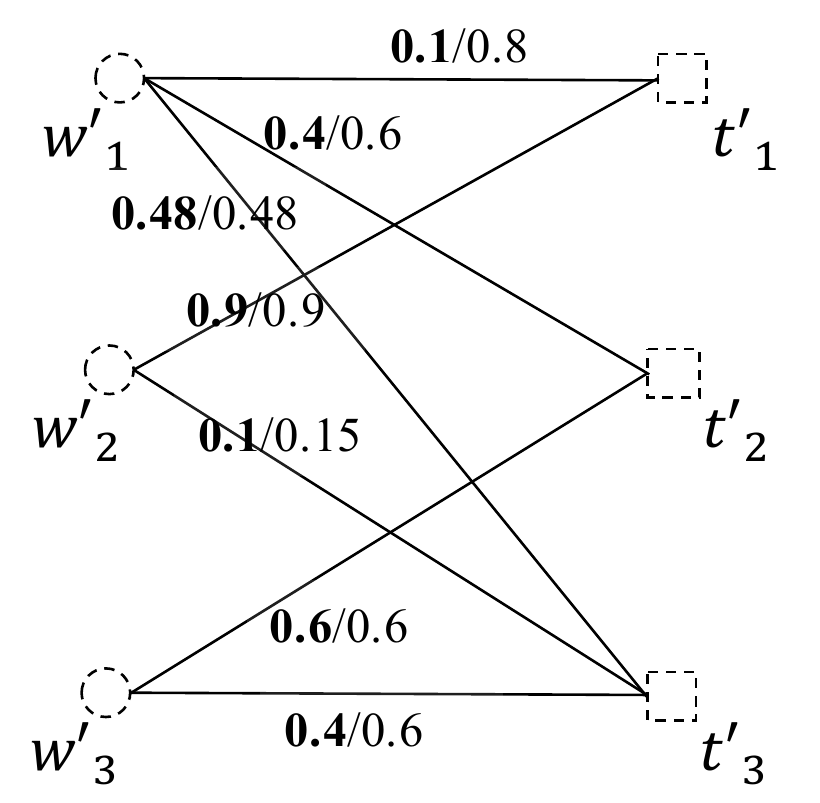}}
		\label{subfig:orr_2_left}}
	\subfigure[][{\scriptsize  The matching }]{
		\scalebox{0.4}[0.4]{\includegraphics{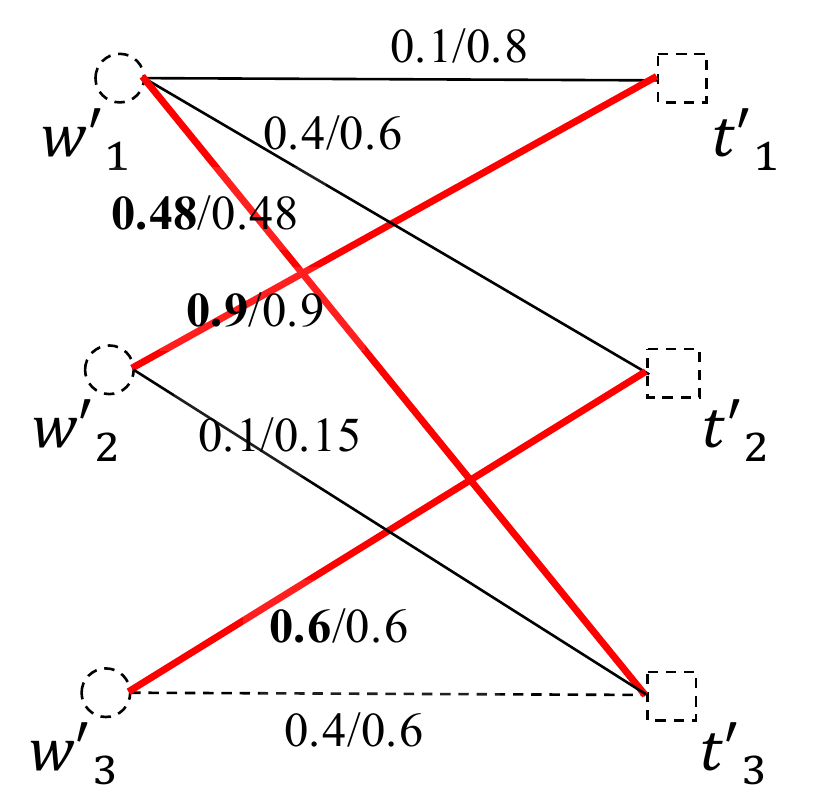}}
		\label{subfig:orr_2_right}}\figureCaptionMargin
    \caption{\small  The maximum weight flow and the rounded matching. }
    \figureBelowMargin\vspace{0ex}
	\label{fig:rr_2}
\end{figure}

\section{k-Switch}
\label{sec:kSwitch}

We present the main contribution of this paper -- our proposed $k$-Switch in this section. We introduce the basic idea of the entire solution first, and then present the technical details in each step of the method. 

\subsection{Overview}

$k$-Switch method is a novel task swapping method, which uses coordination between workers to achieve utility gain. It trades off modest system overhead with secure computation between workers to increase the number of successfully assigned tasks, without privacy leak. 

\begin{figure}[t!]\centering \figureTopMargin\vspace{0ex}
	\scalebox{0.3}[0.3]{\includegraphics{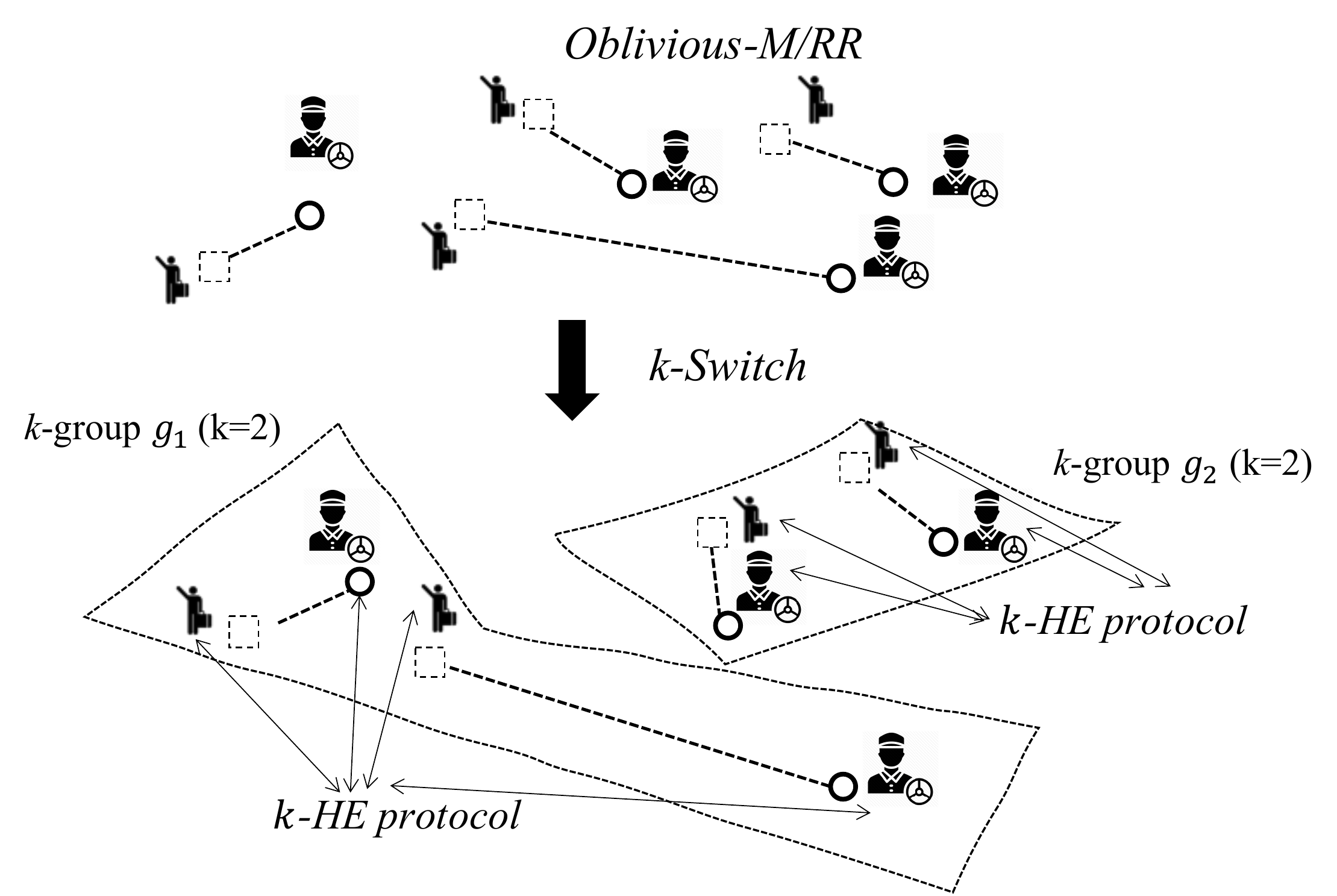}}\vspace{-1ex}
	\caption{\small $k$-Switch}
	\figureBelowMargin\vspace{0ex}

	\label{fig:kSwitch}
\end{figure}

Figure~\ref{fig:kSwitch} shows the basic idea of the method. $k$-Switch starts with a baseline matching obtained by our proposed Oblivious baselines, and then makes an improvement on the matching with the following steps: 

Step 1) Grouping (Sec.~\ref{subsec:grouping}): carefully groups workers into groups of size $k$, where $k$ is a small number selected by the server. 

Step 2) $k$-HE (Sec.~\ref{subsec:kHE}): for each $k$-group, the workers and the task inside the group uses secure computation to calculate the true distance based on encrypted true locations. Then, workers inside the $k$-group swap tasks if such swapping achieves utility gain, i.e., increase the number of successfully assigned tasks. 

Step 3) $\lambda$-Opting (Sec.~\ref{subsec:lambdaOpting}): server iterates Step 1) and Step 2), each time grouping workers into small $k$-groups and letting them use $k$-HE to communicate inside the group to achieve utility gain. 

We introduce the details of each step in the following sections. We show Step 1) Grouping requires solving an NP-hard problem when $k \geq 3$, and our solution includes an efficient greedy grouping algorithm. 

\subsection{Baseline matching}

$k$-Switch starts with a baseline matching $M_0$ obtained by either Oblivious-M or Oblivious-RR proposed in Section~\ref{sec:problem}. According to the experiment results, we find that Oblivious-M outperforms Oblivious-RR in terms of the number of successfully assigned tasks, so we use it in $k$-Switch. In the future, if other better Oblivious matching methods are proposed, they could also be incorporated to find the baseline matching $M_0$.

For clarify of presentation, we formally define the baseline matching $M_0$. 

\begin{definition} (Baseline matching) For a PBTA problem as defined in Def.~\ref{def:pbta}, we use Oblivious-M to find a baseline matching $M_0$: 
	$$M_0 = \{(w, t) | w \in W, t \in T\}$$
	\vspace{-5ex}

	\label{def:m0}
\end{definition}

Again, to illustrate our idea clearly, we continue with the running example in Figure~\ref{fig:example_3}. Using Oblivious-M, we obtain a baseline matching $M_0 = \{(w_1, t_2), (w_2, t_1), (w_3, t_3)\}$. At this stage, only perturbed locations shown in Figure~\ref{fig:example} is observable to the server. 

As we illustrate previously, this baseline matching could be  sub-optimal when we use true locations to verify the reachability constraints of the workers. In the next step Grouping, we divide workers into small groups of size $k$ for further optimization. 

\subsection{Grouping}
\label{subsec:grouping}

\subsubsection{Definitions and k-Grouping problem}

The purpose of grouping is to group nearby workers and tasks into small groups, such that there are good chances that swapping tasks inside the small group could lead to improvement of the number of successfully assigned tasks (the utility of the matching). The small groups are referred as $k$-groups. Next, we formally define the $k$-group and the $k$-Grouping problem. 

For clarify of presentation, we use $P = W \times T$ to denote the set of all possible matched worker/task pairs. For example, in the baseline matching $M_0$ in our running example, we have three matched pairs, $p_1 = (w_1, t_2) \in P$. $p_2 = (w_2, t_1) \in P$, $p_3 = (w_3, t_3) \in P$. The baseline matching could also be denoted as $M_0 = \{p=(w, t) |  w \in W, t \in T\}$.

\begin{definition} ($k$-group) A $k$-group is a set of $k$ matched worker-task pairs in the baseline matching. A $k$-group $g$ contains $k$ different workers and their matched tasks: 
	\vspace{-1ex}
	$$g = \{p =(w,t) | p \in M_0\} \quad s.t. \quad  |g| = k.$$
	\label{def:kGroup}
\vspace{-5ex}

\end{definition}

Back to the running example, if $k=2$, then a $k$-group becomes a $2$-group which contains 2 matched pairs. For example, $g_1=(p_1, p_2)$ is a possible $2$-group, $g_2 = (p_1, p_3)$ is another possible $2$-group. For 2-group $g_1$, it contains worker from the matched pair $p_1$, which is $w_1$. It also contains the worker from the matched pair $p_2$, which is $w_2$. The 2-group $g_1$ also contains the tasks that are matched with the workers, which are task $t_2$ in pair $p_1$ and $t_1$ in pair $p_2$.

The purpose of grouping is to group nearby workers and tasks together such that there are good opportunities that swapping tasks between them leads to better assignment. Some $k$-groups are better than the others. To evaluate the quality of each $k$-group, we propose a simple yet effective  measurements, verified by experiments, called obfuscation-score (short as OScore). We first define the OScore on a $2$-group and then extend it to $k$-group. 

\begin{definition} (Obfuscation-score) Given a $2$-group $g = \{p_1, p_2\}$, we define the obfuscation-score (or OScore) of the $2$-group as: 
	$$\textit{OScore}(g) = d(l'_{p_1.w}, l'_{p_2.w}) + d(l'_{p_1.t}, l'_{p_2.t}).$$
	\label{def:oscore}
	Here $p_1$ and $p_2$ are the two matched pairs inside the $2$-group $g$. Notation $p.w$ and $p.t$ respectively denote the worker and the task from the matched pair. 
\end{definition}

The OScore definition is straightforward, it measures two distance: i) the distance between the perturbed locations of the two workers in the $2$-group; and ii) the distance between the perturbed locations of the two tasks that are assigned to the two workers in the baseline matching $M_0$. Then OScore is the sum of the two distance. 

Intuitively, if both workers and their assigned tasks appear to be close to each other, they are spatially \textit{clustered}, and should be grouped together in a $k$-group. Other measures more sophisticated than OScore could be defined and used, however we adopt this OScore as it is very efficient to compute and yet effective, as to be demonstrated in our experiments. Next, we extend the OScore measurement from $2$-group to $k$-group. 

\begin{definition} (OScore of a $k$-group) For a $k$-group $g = \{p_1, \ldots, p_k\}$, we define its OScore as:
	$$\text{OScore}(g) = \sum_{1 \leq i < j \leq k} \text{OScore}(\{p_i, p_j\}).$$ 
	\label{def:koscore}
	\figureBelowMargin\vspace{-2ex}

\end{definition} 

For a $k$-group $g$, its OScore is the sum of the OScore of $g$'s subsets of size 2. We abuse notations to use the same OScore to refer to different definitions for $2$-group and $k$-group when $k\geq 3$. 

Before we introduce the $k$-grouping problem, we define the term $k$-division as a collection of $k$-groups that we select out of the baseline matching $M_0$. 

\begin{definition} ($k$-division) Given a baseline matching $M_0 = \{p | p \in P\}$, a $k$-division $D = \{g_1, \ldots, g_d\}$ is a set of non-overlapping $k$-groups of $M_0$, where the size of the $k$-division is $d = \lceil |M_0|/k\rceil$. We define the score of a $k$-division as:
	$$\text{score}(D) = \sum_{1 \leq i \leq d} \text{OScore}(g_i).$$
	\label{def:k-division}
	Remark: i) for $1 \leq i \leq d$, $|g_i|=k$ except for at most one group, when $|M_0|$ is not a multiple of $k$; ii) $k$-groups are non-overlapping (disjoint) when workers from any $k$-group is different from any other $k$-groups. 
\end{definition}

The $k$-division is a collection of $k$-groups from the baseline matchings. It divides workers into small $k$-groups, each with exact size $k$, except for at most one group. The exception happens when the total number of workers is not a multiple of $k$, and the exceptional group has a size in the range of $[0, k-1]$. The score of the $k$-division is defined as the sum of OScore of all its $k$-groups.  

Next we are ready to formally define the $k$-Grouping problem. It divides the workers in the baseline matching $M_0$ into small groups to create a $k$-division as defined in Def.~\ref{def:k-division}, and minimizes the score of the $k$-division. 

\begin{definition} ($k$-Grouping problem (KGP)) Given a matching $M_0 = \{p | p \in P\}$, the $k$-grouping problem returns a $k$-division $D^*$ with minimal score. Formally, for any other $k$-division $D$: 
	$$\text{score}(D^*) \leq \text{score}(D).$$  
	\figureBelowMargin\vspace{-2ex}

	\label{def:KGP}
\end{definition}

Next, we provide theoretical analysis on the hardness of $k$-Grouping problem (KGP). We show for $k=2$, KGP is polynomial-time solvable, while for $k\geq 3$, it is hard to approximate. 

\subsubsection{Algorithms for $k=2$}

\begin{theorem} 
	When $k=2$, KGP is in P class, solvable in polynomial time.
	\label{theo:kpgk2}
\end{theorem}
\vspace{-2ex}
\begin{proof}
	The overall idea of the proof is that, when $k=2$, the KGP is equiv. to finding the maximum weight matching on a general graph. We construct a graph $M_\mathcal{G}$ from the baseline matching $M_0 = \{p | p \in P\}$ as follows. For each matched pair $p \in P$, we create a vertex $p$ and add it to graph $M_\mathcal{G}$. We create an edge between any two vertices $p$. It is a complete graph. 
	We define the weight of the edges as follows. For two vertices $p_i \in M_0$ and $p_j \in M_0$, we set the edge weight $w(p_i, p_j) = $ OScore$(\{p_i, p_j\})$, where OScore$(\{p_i, p_j\})$ is defined in Def.~\ref{def:oscore}. 
	
	We could then show that the maximum weight matching on $M_\mathcal{G}$ corresponds to the optimal solution to the KGP when $k=2$. We defer the details of the proof to the appendix (\secref{subsec:appendix_kswitch}).
	\vspace{-2ex}

\end{proof}

\begin{theorem} 
	When $k=2$, KGP is solvable with time complexity $\mathcal{O}(|M_0|^{2.37})$.
	\label{theo:kpgk2_time}
\end{theorem}

\begin{proof}
	For the maximum weight matching problem on a general graph, it is shown that it is among the \textit{hardest} problem that could be solved in polynomial time, with $O(|V|^{3})$ time complexity \cite{edmond1965}, where $V$ denotes the vertex set of the graph. In our setting, $|V| = |M_0|$, the size of the baseline matching. 
	
	In our setting, as the constructed graph $M_\mathcal{G}$ is a complete graph, and convertible to bipartite graph via a simple \textit{graph transformation} method (please refer to the appendix, \secref{subsec:appendix_kswitch}). There exists matrix multiplication algorithms on the transformed bipartite graph to obtain the maximum weight matching with time complexity $O(|V'|^{2.37})$. Because the graph transformation doubles the number of vertices of the original graph $M_\mathcal{G}$, we have $|V'| = 2 * |V| = 2|M_0|$. Overall, KGP problem is solvable in $\mathcal{O}(|V'|^{2.37}) \to \mathcal{O}(2^{2.37} * |V|^{2.37}) \to \mathcal{O}(|M_0|^{2.37})$. If we let $n=\max(|W|, |T|)$, because $|M_0|$ is bounded by $n$, then KGP is solvable in $ \mathcal{O}(n^{2.37})$. 
	\vspace{-2ex}

\end{proof}

\subsubsection{Algorithms for $k \geq 3$}

\begin{theorem} When $k \geq 3$, KGP has no polynomial time approximation algorithm with finite approximation ratio unless P=NP. 
	\figureBelowMargin\vspace{-2ex}

	\label{theo:kpgk3}
\end{theorem}

\begin{proof}
	We show a polynomial reduction of the \textit{Perfectly Balanced Graph Partition} (PBGP) problem to the KGP. Because PBGP has no polynomial approximation algorithm with finite approximation ratio unless P=NP \cite{DBLP:journals/mst/AndreevR06}, our KGP has the same hardness. 
	
	We review PBGP: given a graph $G=(V,E)$, with weight $w(e)$ on each edge $e$. For an integer $p \geq 2$, a $p$-partition is $p$ disjoint subsets with equal sizes: $V=V_1 \cup \cdots \cup V_p$. We assume $|V|$ is a multiple of $p$ here, so $|V_1| = \cdots |V_p| = |V|/p$. The decision version of PBGP is that given a positive integer $W$, is there a $p$-partition such that for the cross edges set $E' \subset E$, which have two endpoints in the two different sets $V_i$, i.e., $E' = \{(v_i, v_j)\in E | v_i \in V_i, v_j \in V_j , i \neq j \}$, the sum of the weights of edges in $E'$ is less or equal to the given integer $W$, i.e., $\sum_{e \in E'}w(e) \leq W$?
	
	Please refer to the appendix (\secref{subsec:appendix_kswitch}) for the details of the proof. Because we could show that PBGP  $\leq _{p}$ KGP, and since PGBP has no polynomial time approximation algorithm with finite approximation factor unless P=NP \cite{DBLP:journals/mst/AndreevR06}, KGP has the same hardness. 
\vspace{-2ex}
\end{proof}

Because of the intractability of KGP, we propose an efficient greedy method to find a $k$-division. The greedy algorithm \textit{packs} $k$-group one by one, each time starting with an empty set, and keeps adding a new $2$-group $\{p_i, p_j\}$ with smallest OScore to the current $k$-group. If $k$ is an odd number, it randomly picks the last matched pair $p$. The detailed steps is shown in Alg.~ \ref{algo:greedy_grouping}. 

\begin{algorithm}[t]
	\DontPrintSemicolon
	\KwIn{ A baseline matching $M_0$. Perturbed locations $l_{w'}$ and $l_{t'}$ for each $w\in W$ and $t \in T$. }
	\KwOut{A $k$-division $D$.}

	Initialize a heap $h$, an empty set $D$ \;
	\ForEach{$p_1 \in M_0$}{
		\ForEach{$p_2 \in M_0$}{
			\If{$p_1.w \neq p_2.w$}{
				oscore$:= d(l'_{p_1.w}, l'_{p_2.w}) + d(l'_{p_1.t}, l'_{p_2.t})$ \;
				$h$.insert($\{p_1, p_2\}$, oscore) \;
				\label{algo:line:heap}
			}
		}
	}

	$d := \lceil |M_0| / k \rceil$ \tcp*{Calculate how many $k$-groups} 
	\label{algo:line:kgroups}

	\For{$i:=0$; $i < d$; $i=i+1$}{
		$g:=\{\}$ \;
		\While{$g.size() < k-1$}{
			\label{algo:line:forloop_start}

			$p_1, p_2:=h$.pop() \;
			$g$.Insert($p_1$, $p_2$) \;
			Mark $p_1.w$ or $p_2.w$ as used\;
			\label{algo:line:forloop_end}

		}
		$D$.Insert($g$) \;
	}
	\Return{$D$}\;
	\caption{\texttt{Greedy-Grouping} }
	\label{algo:greedy_grouping}
\end{algorithm}

At Line~\ref{algo:line:heap}, the algorithm uses a heap storing all combinations of pairs in the baseline matching $\{p_i, p_j\}$ with its OScore as the sorted key. Every time we could pop the pair $\{p_i, p_j\}$ with the smallest OScore and add them to the current $k$-group. After each $k$-group is formed, we continue to the next one until a $k$-division is obtained and returned. Line~\ref{algo:line:forloop_start}-\ref{algo:line:forloop_end} greedily add two matched pairs to the current $k$-group $g$, until its size reaches $k$. In total, we form $d$ groups, as calculated at Line~\ref{algo:line:kgroups}. We omit some details for checking used workers and randomly picking the last item (see full version in the appendix, \secref{subsec:appendix_kswitch}). 

The time complexity of Alg.~ \ref{algo:greedy_grouping} is dominated by the heap construction, which takes $\mathcal{O}(e\log e)$, where $e$ is the total number of elements inserted to the heap. We know $e$ is all the combinations of matched pairs, as shown at Line 2-3, so $e=\mathcal{O}(n^2)$, where $n=\max(|W|,|T|)$, the size of the baseline matching. In conclusion, for Alg.~\ref{algo:greedy_grouping} has an $\mathcal{O}(n^2\log n^2) \to \mathcal{O}(n^2\log n)$ time complexity.  

\subsection{k-HE protocol}
\label{subsec:kHE}

The purpose of the previous step Grouping is to divide all workers into small groups of size $k$. In this section, we introduce how workers inside each small $k$-group utilize secure computation in parallel to increase the number of successfully assigned tasks via task swapping. 

Our $k$-HE protocol runs in the small group of size $k$. The secure computation is based on the Paillier Crypto-system \cite{DBLP:conf/eurocrypt/Paillier99} in Homomorphic Encryption (HE), which allows Homomorphic Addition and Homomorphic Multiplication. Similar to the global task assignment setting in HESI framework \cite{DBLP:conf/edbt/LiuCZZZQ17}, we also use HE for secure distance calculation. In contrast, our protocol is restricted to small size $k$ (ranging from 2 to 8). We allow workers and tasks inside the small group to communicate the encrypted true locations with one another, and if the number of successfully assigned tasks could be improved based on their true locations, then workers swap tasks between themselves. 

Figure \ref{fig:k-HE} gives an illustration of the protocol. Two entities (either worker or task) out of the $k$-group are randomly elected and serve as the proxy servers $P_a$ and $P_b$. Then, we perform secure distance calculation between each pair of workers and tasks following the major steps of HESI \cite{DBLP:conf/edbt/LiuCZZZQ17} (details deferred to the appendix, \secref{subsec:appendix_kHE}).  As for the time complexity, two proxy servers enumerate all combinations of worker-task pairs in the $k$-group and compute the true distance, with $\mathcal{O}({k \choose 2})\to \mathcal{O}(k^2)$ time. Then, $P_b$ runs a matching algorithm w.r.t. the true distances, with $O(k^{3})$ time using the max-flow algorithm similar to the Oblivious-M baseline (Sec. ~\ref{subsec:om}). The overall time complexity is $O({k \choose 2} + k^{3}) \to O(k^{3})$. 

\begin{figure}[ht!]\centering \figureTopMargin\vspace{0ex}
	\scalebox{0.4}[0.4]{\includegraphics{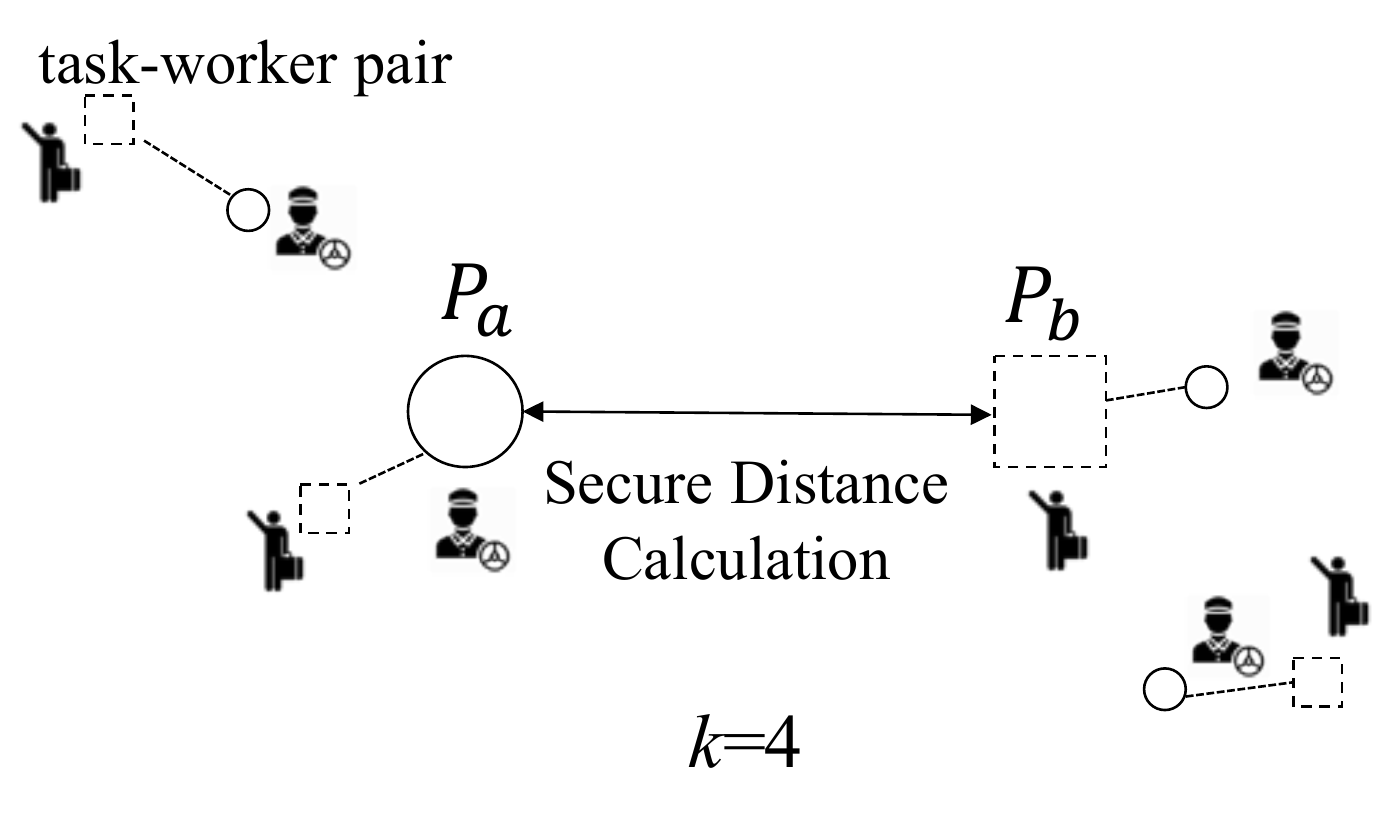}}\figureCaptionMargin
	\caption{\small k-HE protocol (the case for $k=4$)}
	\figureBelowMargin\vspace{0ex}
	\label{fig:k-HE}
\end{figure}

\subsection{$\lambda$-Opting}
\label{subsec:lambdaOpting}

The previous section describes how small $k$-groups execute $k$-HE protocol in parallel and workers swap tasks if task swapping increases the number of successfully assigned tasks within the group. The last phase of $k$-Switch is $\lambda$-Opting, which iterates Grouping and $k$-HE protocol for $\lambda$ rounds. It also stops if no utility gain is obtained at the current round. $\lambda$ is a system parameter controlling the trade-off between utility gain and system overhead.

\begin{figure}[h!]\centering \figureTopMargin\vspace{1ex}
	\scalebox{0.4}[0.4]{\includegraphics{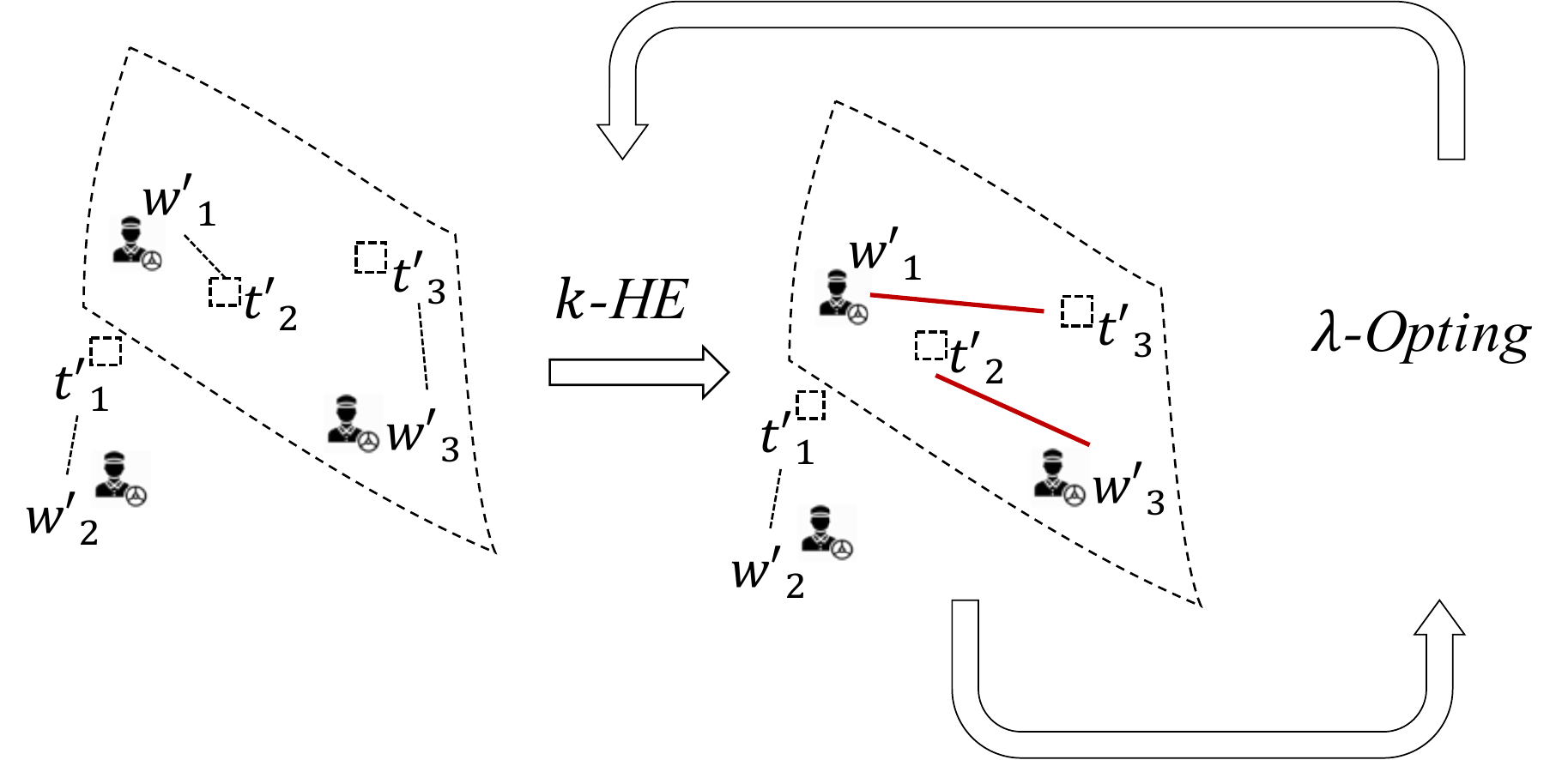}}
	\caption{\small $\lambda$-Opting.}
	\label{fig:lambda_opt}
	\figureBelowMargin\vspace{0ex}

\end{figure}

Each round of $\lambda$-Opting executes the grouping and the $k$-HE protocol, taking $\mathcal{O}(k^{3} + n^2\log n)$ time. The time complexity of $k$-Switch is thus $\mathcal{O}(\lambda k^{3} + \lambda n^2\log n)$, where $n = \max(|W|, |T|)$.

\section{Experimental Study}
\label{sec:experiment}

We conduct extensive experiments on both the real-world dataset and the synthetic dataset to validate the effectiveness and efficiency of our proposed $k$-Switch method. 

With respect to the effectiveness, as measured by the number of successfully assigned tasks, experiments show that $k$-Switch assigns up to 5.9$\times$ more tasks than other batch-based baselines, and assigns up to 1.74$\times$ more tasks than the competing online method SCGuard. In terms of efficiency, as measured by the running time, our method is efficient, finishing within 1.5 minutes on datasets of moderate sizes (500 workers and 500 tasks). While being slightly slower than other methods (slower than SCGuard by about a constant factor of 2), $k$-Switch is considered cost-effective because it trades off minor system overhead with considerable utility gain. 

\subsection{Experimental setup}
\label{sec:experimentSetup}

\subsubsection{Datasets}

We conduct the experiments on both the real-world and the synthetic dataset. The real-world dataset is the taxi dataset from Didi Chuxing \cite{didi2018}. For the synthetic dataset, we randomly sample workers and tasks' locations from the range $[0, 8000] \times [0,8000]$.
\vspace{-1ex}
\subsubsection{Baselines}

The baselines we test include the two baseline solutions we propose: \textbf{Oblivious-M} (short as OM, introduced in Sec.~\ref{subsec:om}) and \textbf{Oblivious-RR} (ORR, Sec.~\ref{subsec:orr}). In addition, we test \textbf{SCGuard} (SCG, \cite{to2018}), an online method allowing each newly arrived task to interactively check several other workers to see whether the task could be assigned. 
\vspace{-1ex}
\subsubsection{Metrics and control variables}

\noindent\textbf{Control variables}: 
Number of workers $w \in [\textbf{100}, 200, 500, 1000]$. 
Number of tasks $t \in [\textbf{100}, 200, \\ 500, 1000]$. System parameter $k \in [\textbf{2}, 4, 6, 8]$, $\lambda \in [5, 10, \textbf{20}]$. Privacy requirement $\epsilon \in [\textbf{0.4}, 1.25, 2.5]$. We set $r=1000$m as constant. Varying $\epsilon$ corresponds to varying Geo-I privacy level $l = \epsilon r \in [\textbf{400}, 1250, 2500]$. Default parameters are in boldface. 

\noindent\textbf{Metrics}. We focus on 1) effectiveness (\textit{utility}), measured by the number of tasks assigned, and 2) efficiency, measured by running time in seconds. 

\noindent\textbf{System configuration}. The experiments are performed on a MacBook Pro with 1.4GHz Quad-Core Intel Core i5 and 16GB 2133MHz LPDDR3 memory, running MacOS 11.0. The methods were implemented in Python. \vspace{-1ex}

\subsection{Experimental results}
\label{sec:experimentResults}

\subsubsection{Effectiveness}

\noindent\textbf{Overview:} $k$-Switch (short as KS in the figures) outperforms other baselines by significant margins, over different privacy levels (Fig.~\ref{fig:exp_1_utility}), across datasets of different sizes (Fig.~\ref{fig:exp_3_utility_size}). For a stricter privacy parameter $\epsilon=0.4$ on the taxi dataset (shown in Fig.~\ref{subfig:exp_1_utility_real}), $k$-Switch achieves 5.9$\times$ improvement over baseline Oblivious-M, and 1.74$\times$ improvement over the competing online method SCGuard.

\begin{figure}[t!]\centering \figureTopMargin\vspace{-1ex}
	\subfigure[][{\scriptsize Real-world dataset}]{
		\scalebox{0.23}[0.23]{\includegraphics{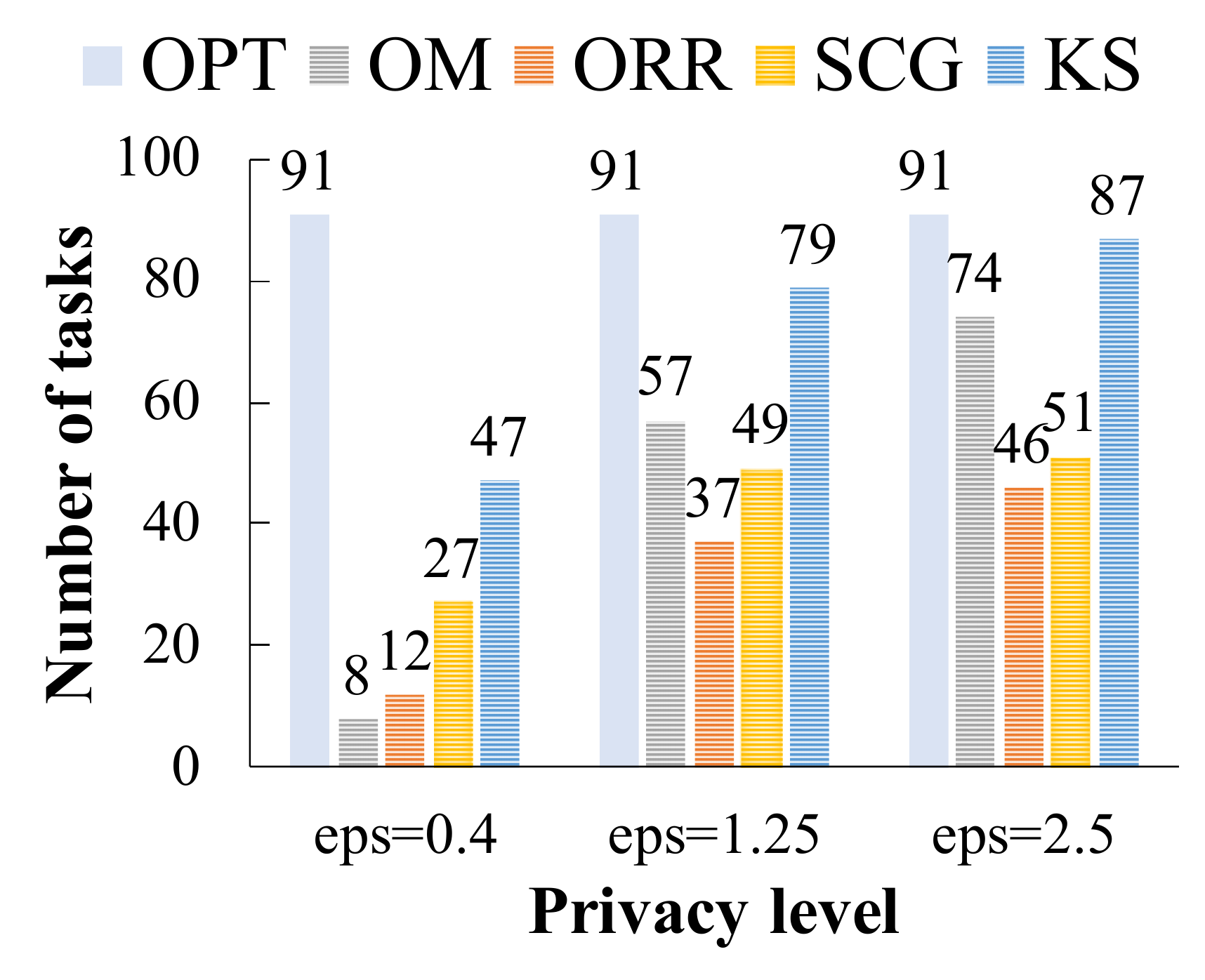}}
		\label{subfig:exp_1_utility_real}}
	\subfigure[][{\scriptsize  Synthetic dataset}]{
		\scalebox{0.23}[0.23]{\includegraphics{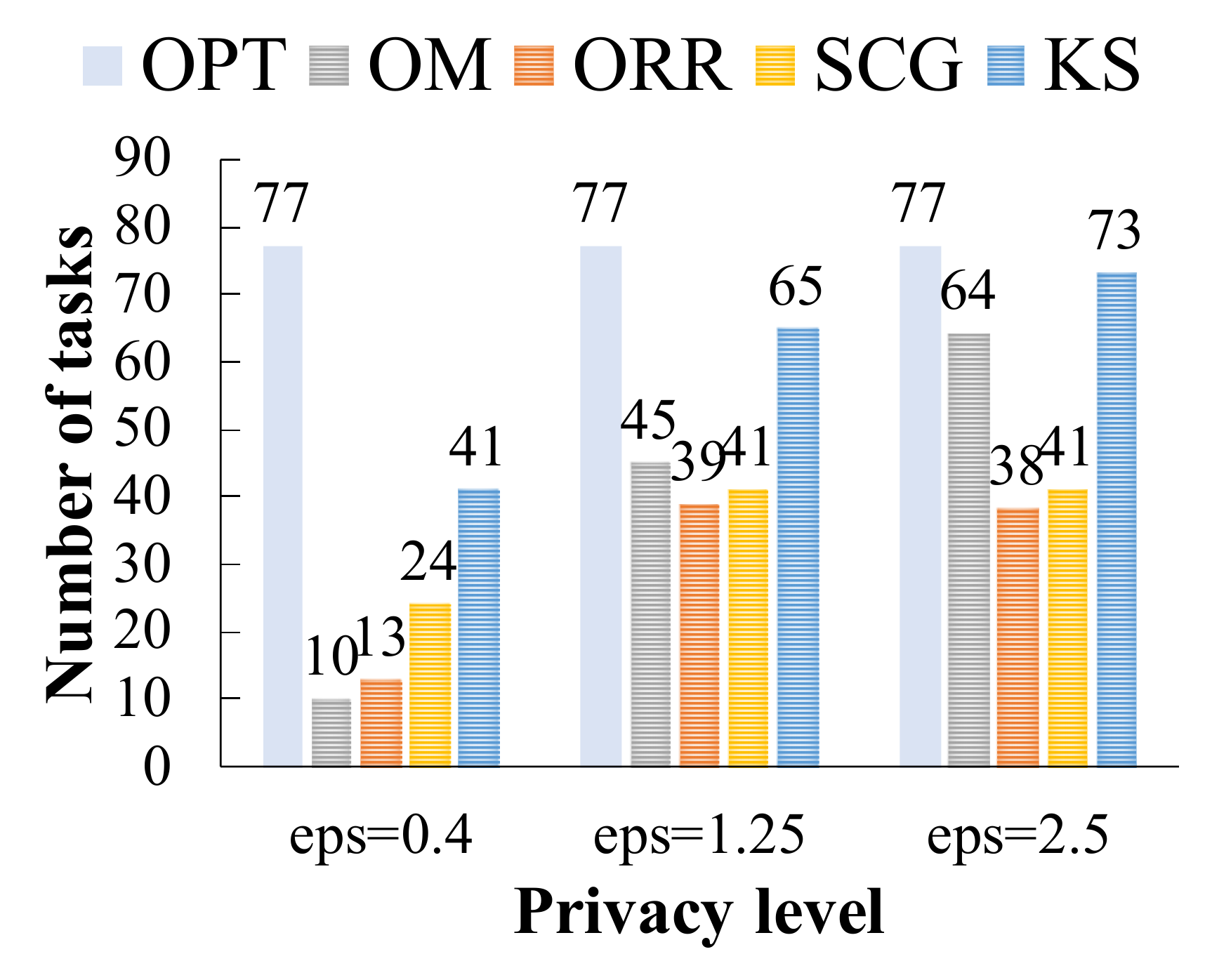}}
		\label{subfig:exp_1_utility_syn}}\figureCaptionMargin
    \caption{\small Number of tasks assigned for different methods, over different privacy levels, 100 workers vs. 100 tasks, $k=2, \lambda=20$.}
    \figureBelowMargin\vspace{0ex}
	\label{fig:exp_1_utility}
\end{figure}

\fakeparagraph{Details of results} Fig.~\ref{fig:exp_1_utility} shows the number of successfully assigned tasks obtained by different methods, over different privacy levels on a dataset of 100 workers vs. 100 tasks. We compare the optimal matching with the Oblivious baseline methods. The optimal matching (OPT) is obtained by using the ground-truth locations, which are not available in the inputs to our PBTA problem. The Oblivious baselines OM and ORR are using only the perturbed locations. At privacy level $\epsilon=0.4$ (shown in Fig.~\ref{subfig:exp_1_utility_real}), the gap between OPT and OM is 82. The OPT is about 11 times larger than OM. This validates the motivation of our research: privacy-preserving techniques perturb the locations of workers and tasks, and directly assigning tasks based on perturbed locations  is erroneous and prone to sub-optimal assignment. The gap is also observable on the synthetic dataset, shown in Fig.~\ref{subfig:exp_1_utility_syn}. 

Varying privacy levels: $k$-Switch outperforms all other methods on different privacy levels, assigning 47 tasks for $\epsilon=0.4$ on taxi data (shown in Fig.~\ref{subfig:exp_1_utility_real}), achieving 5.9$\times$ improvement over OM, which assigns 8 tasks. It also achieves 1.74$\times$ improvement over the SCGuard, which assigns 27 tasks. When $\epsilon$ gets larger, the privacy requirement gets less strict, the gap between the OPT and the oblivious OM and ORR gets smaller. For $\epsilon=2.5$, OM assigns 74 tasks, much closer to the 91 tasks from the OPT solution, as compared to a stricter $\epsilon$. Nevertheless, $k$-Switch delivers strong performance, and assigns 87 tasks, which is close to the OPT. 

\begin{figure}[t!]\centering \figureTopMargin\vspace{-1ex}
	\subfigure[][{\scriptsize Real-world dataset}]{
		\scalebox{0.23}[0.23]{\includegraphics{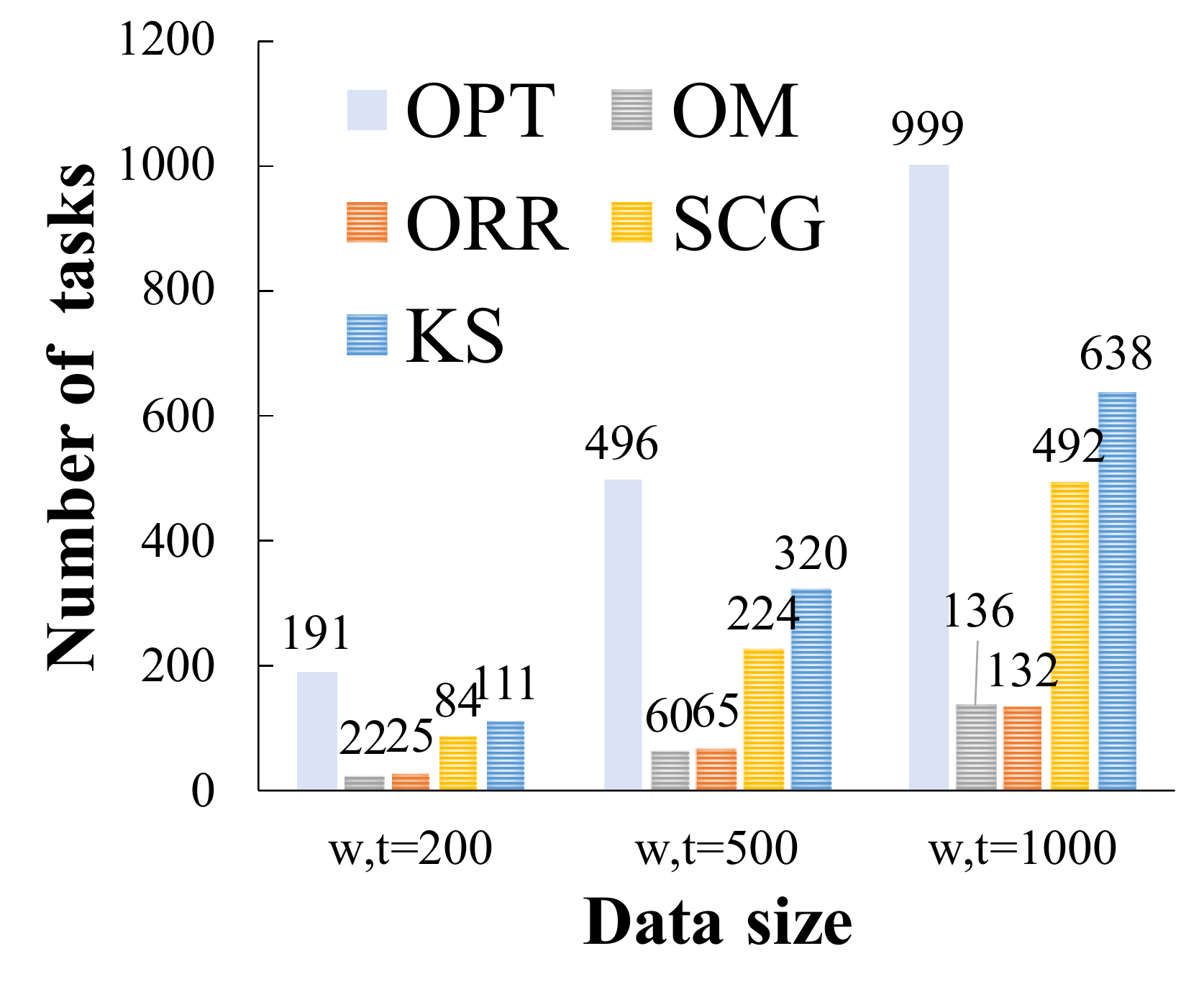}}
		\label{subfig:exp_3_util_size_real}}
	\subfigure[][{\scriptsize  Synthetic dataset}]{
		\scalebox{0.23}[0.23]{\includegraphics{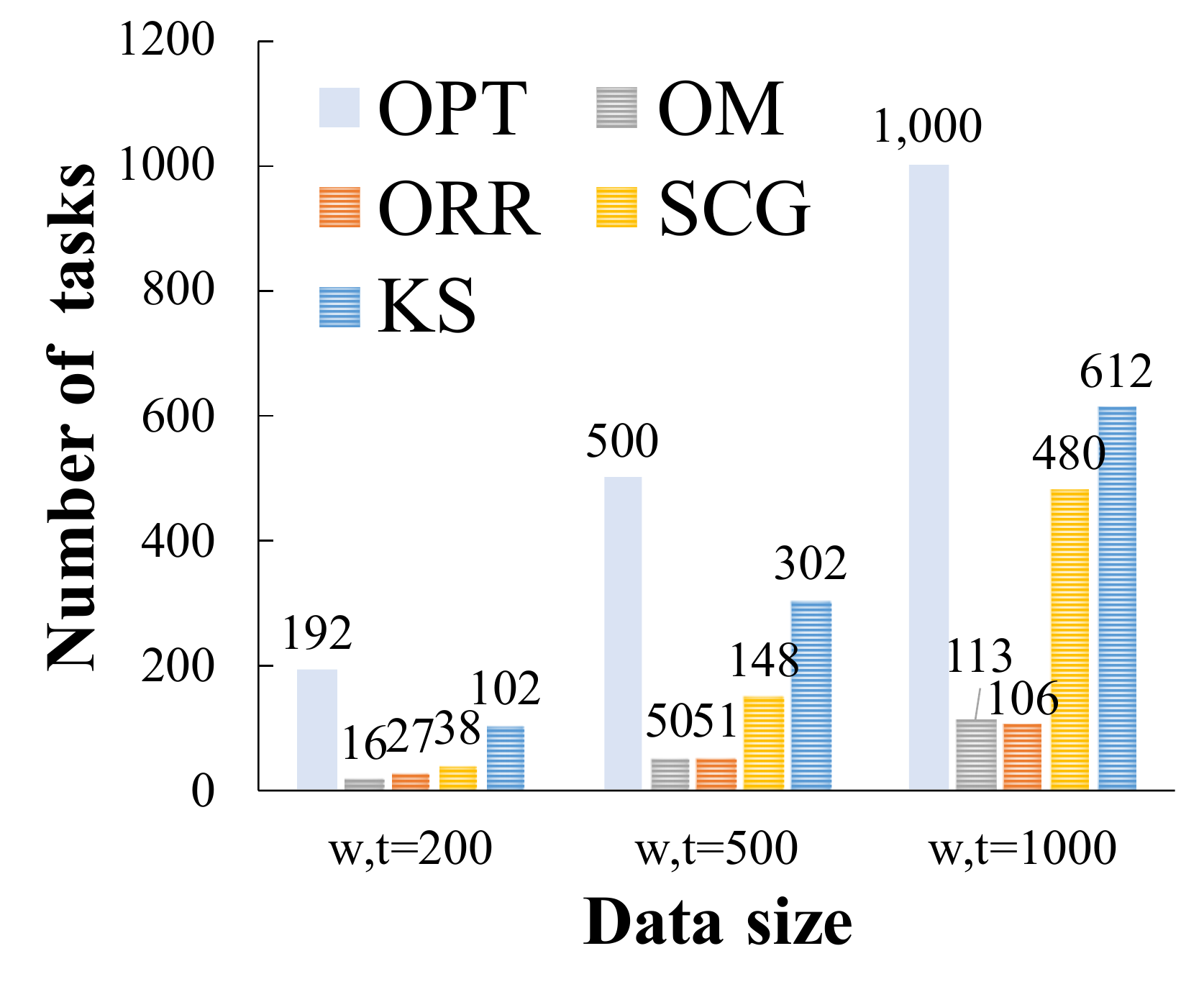}}
		\label{subfig:exp_3_util_size_syn}}\figureCaptionMargin
    \caption{\small Number of tasks assigned, on different input sizes (200-1000), $\epsilon=0.4,k=2, \lambda=20$.}
    \figureBelowMargin\vspace{0ex}
	\label{fig:exp_3_utility_size}
\end{figure}

Varying data size: Fig.~\ref{fig:exp_3_utility_size} shows the number of successfully assigned tasks of different methods, on datasets of different sizes. First, as the number of worker and task increases from 200 to 1000, the OPT result increases from 191 to 999 on taxi dataset (Fig.~\ref{subfig:exp_3_util_size_real}), and 192 to 1000 on the synthetic dataset (Fig.~\ref{subfig:exp_3_util_size_syn}). The gap between OPT and OM is consistently large. For the taxi dataset (Fig.~\ref{subfig:exp_3_util_size_real}), the OPT/OM ratio is $191/22 = 8.7$ for $w, t=200$ and $999 / 136 = 7.3$ for $w, t= 1000$, respectively. $k$-Switch significantly improves over the baseline OM and the online SCGuard method. It achieves 5.05$\times$ improvement over OM for $w,t=200$, 5.33$\times$ for $w,t=500$, and 5.64$\times$ for $w,t=1000$. As for the comparison with SCGuard, $k$-Switch obtains 1.32$\times$, 1.43$\times$, and 1.30$\times$ improvement for $w,t=200$, $500$, and $1000$ respectively. The same behavior is observed on the synthetic dataset (Fig.~\ref{subfig:exp_3_util_size_syn}). 
 
Varying parameter $k$: Fig.~\ref{fig:exp_4_k_utility} shows the number of successfully assigned tasks, over different $k$, on datasets of different sizes (Fig.~\ref{subfig:exp_4_k_utility_size}) and over different privacy requirements ((Fig.~\ref{subfig:exp_4_k_utility_eps})). As Fig.~\ref{subfig:exp_4_k_utility_size} shows, when $k$ increases, the number of successfully assigned tasks increases. For the smallest data size, the utility increases from 47 to 75 tasks, as $k$ increases from 2 to 8. On the other hand, on the data size $w,t=100$, when we vary privacy parameters, the effect of $k$ is not as significant (Fig.~\ref{subfig:exp_4_k_utility_eps}). 

Varying parameter $\lambda$: We defer the results about varying the system parameter $\lambda$ to the appendix (\secref{subsec:appendix_experiments}). The results verify that $\lambda$ achieves a tradeoff between the utility and efficiency for $k$-Switch, which is consistent with our system design. 

\begin{figure}[t!]\centering \figureTopMargin\vspace{0ex}
	\subfigure[][{\scriptsize On different input sizes}]{
		\scalebox{0.23}[0.23]{\includegraphics{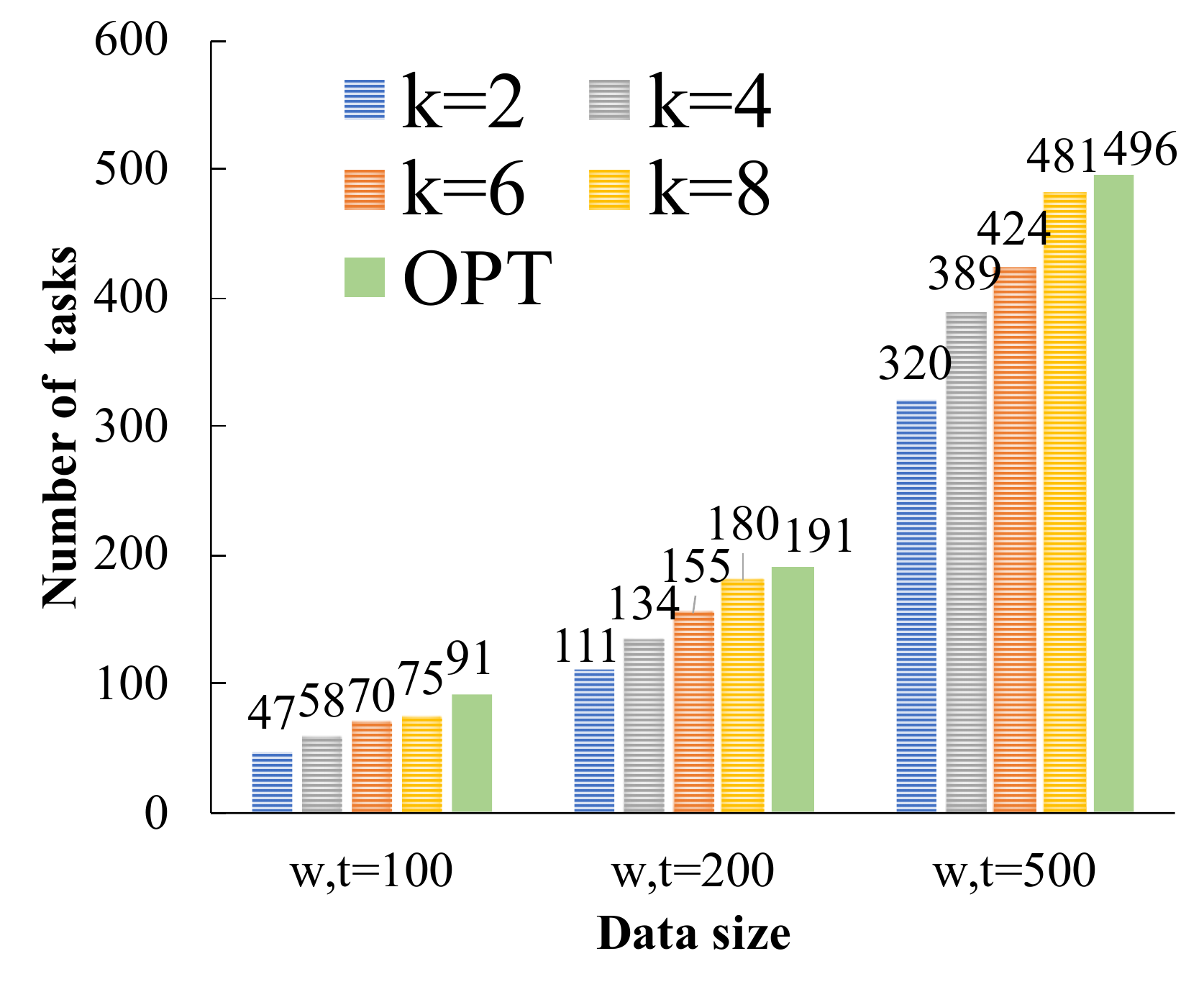}}
		\label{subfig:exp_4_k_utility_size}}
	\subfigure[][{\scriptsize  For different privacy parameter}]{
		\scalebox{0.23}[0.23]{\includegraphics{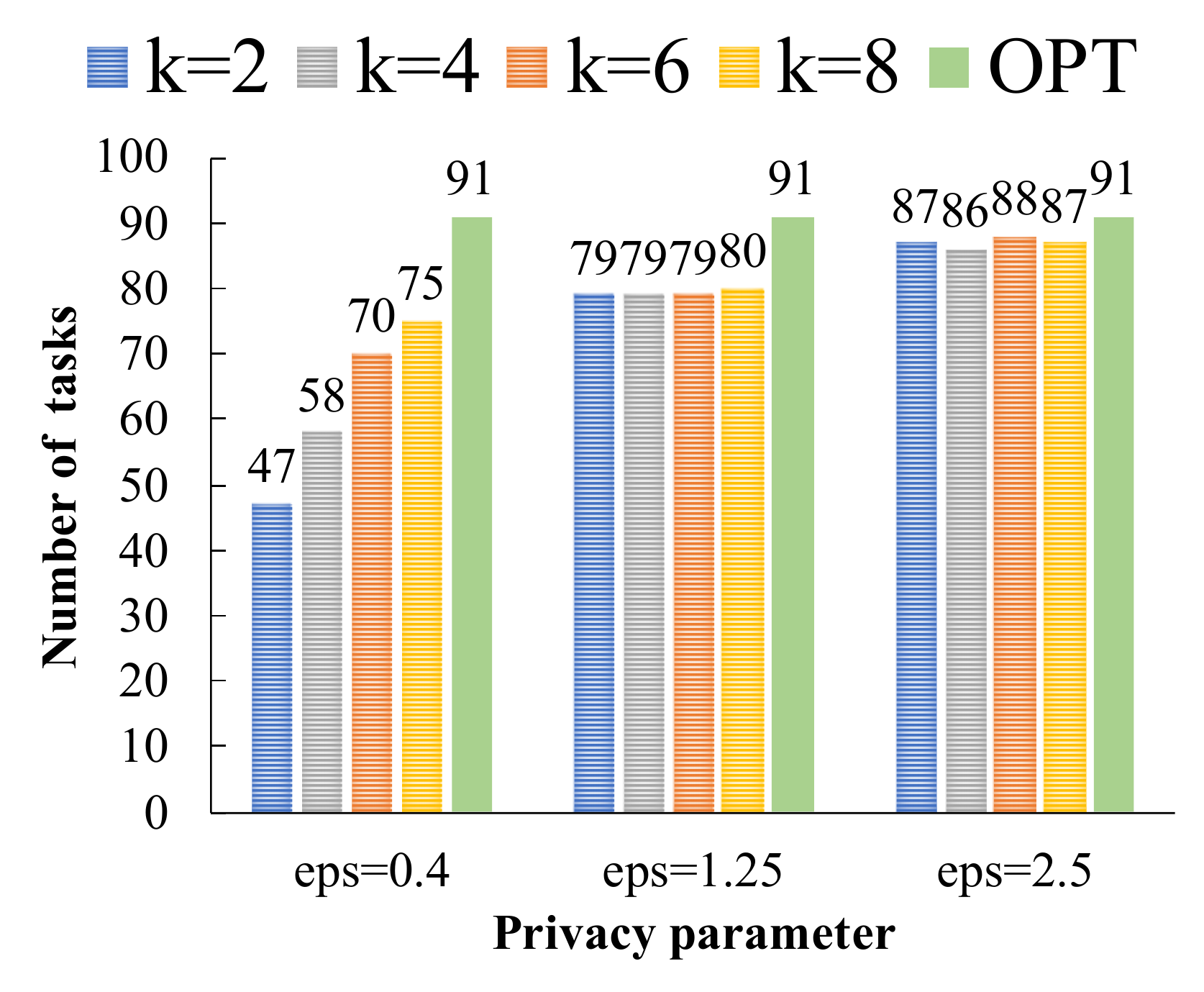}}
		\label{subfig:exp_4_k_utility_eps}}\figureCaptionMargin
    \caption{\small Effect of parameter $k$ for $k$-Switch method on the number of tasks assigned.}
    \figureBelowMargin\vspace{0ex}
	\label{fig:exp_4_k_utility}
\end{figure}

\subsubsection{Efficiency}
\fakeparagraph{Overview} While we expect $k$-Switch to be slower than other methods as the design trades off moderate system overhead with significant utility gain, experimental results show it is only slightly slower than other methods. On our default setting (Fig.~\ref{fig:exp_2_time}), 100 workers vs. 100 tasks, it takes only around 2 seconds to run on a laptop machine. When tested on larger sizes dataset (Fig.~\ref{fig:exp_3_time}), $k$-Switch shows stable efficient running time, a small constant factor slower than the competing SCGuard. Experiments verify our time complexity analysis of the methods. 
\begin{figure}[t!]\centering \figureTopMargin\vspace{-2ex}
	\subfigure[][{\scriptsize Real-world dataset}]{
		\scalebox{0.23}[0.23]{\includegraphics{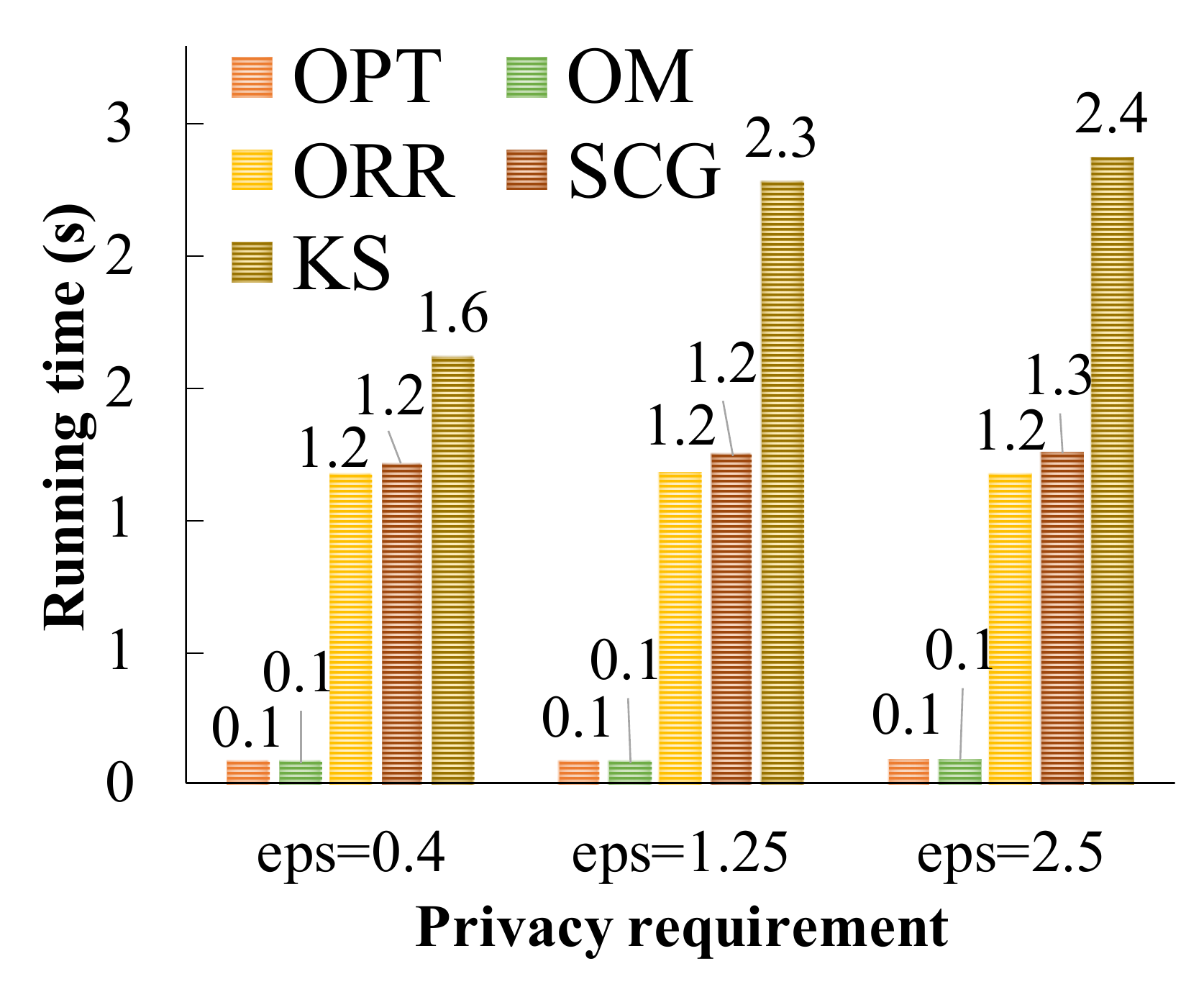}}
		\label{subfig:exp_2_time_real}}
	\subfigure[][{\scriptsize  Synthetic dataset}]{
		\scalebox{0.23}[0.23]{\includegraphics{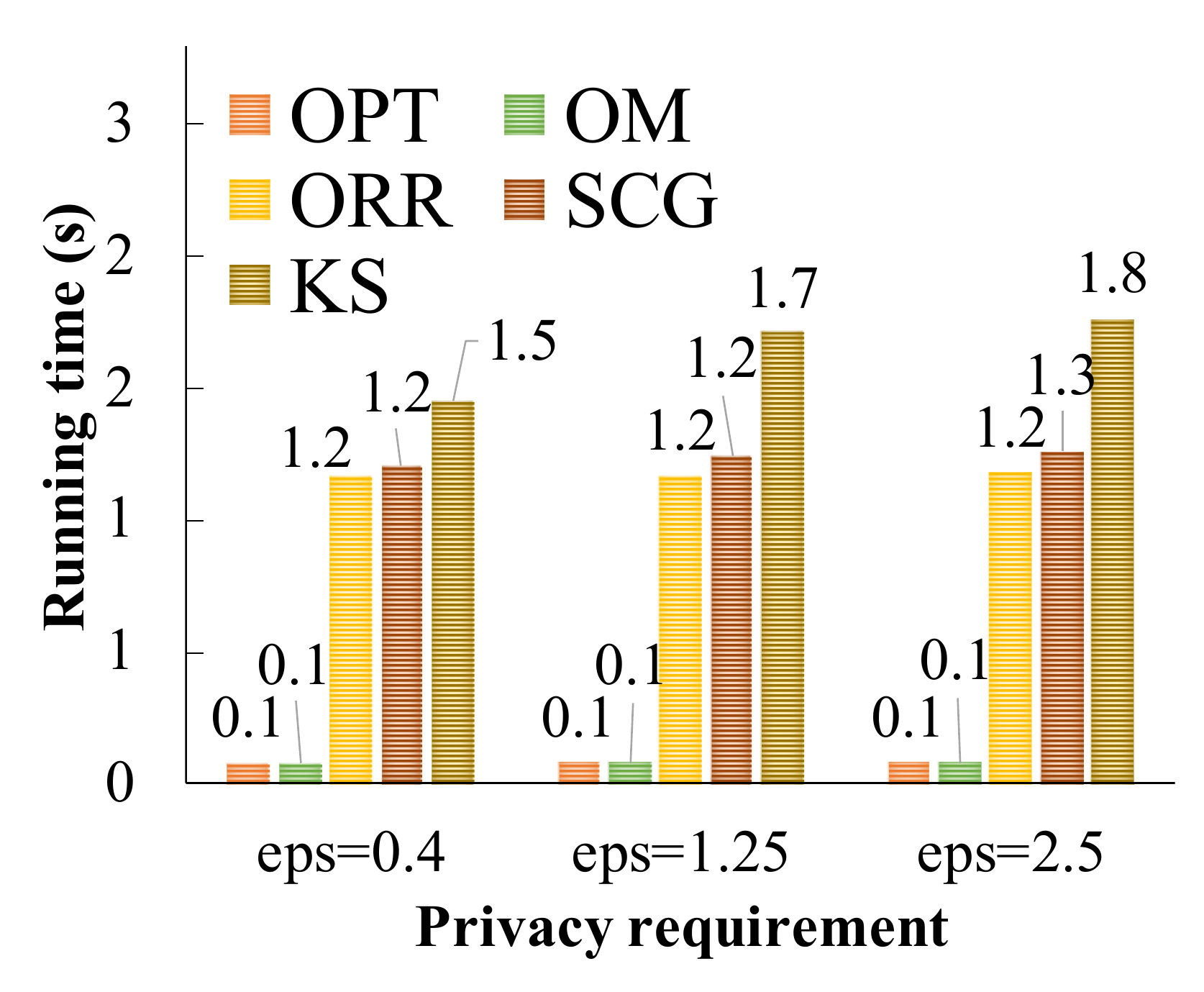}}
		\label{subfig:exp_2_time_syn}}\figureCaptionMargin
    \caption{\small Running time (seconds) for different methods, over different privacy levels, 100 workers vs. 100 tasks, $k=2$.}
    \figureBelowMargin\vspace{0ex}
	\label{fig:exp_2_time}
\end{figure}

\begin{figure}[t!]\centering \figureTopMargin\vspace{0ex}
	\subfigure[][{\scriptsize Real-world dataset}]{
		\scalebox{0.23}[0.23]{\includegraphics{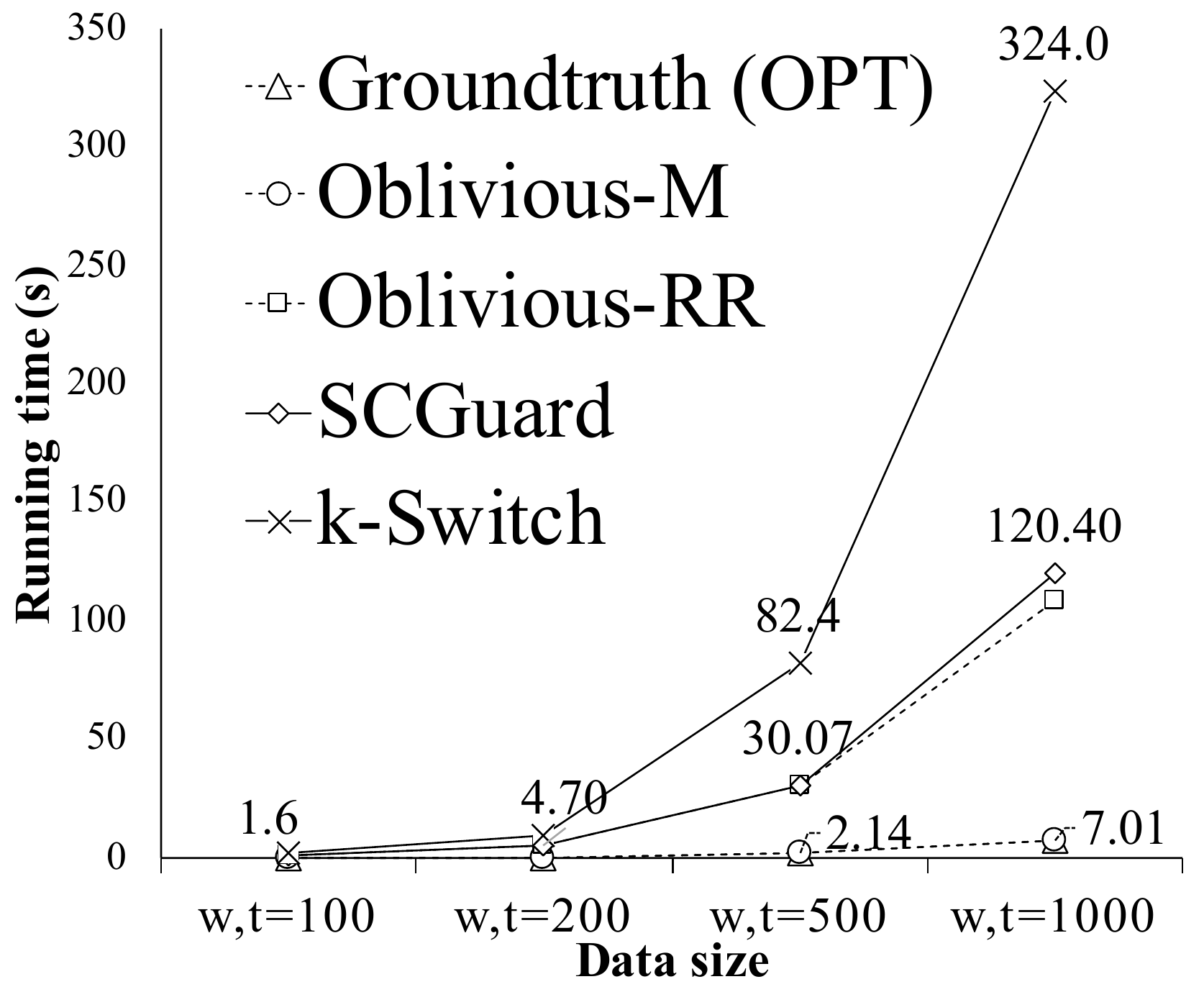}}
		\label{subfig:exp_3_time_real}}
	\subfigure[][{\scriptsize  Synthetic dataset}]{
		\scalebox{0.23}[0.23]{\includegraphics{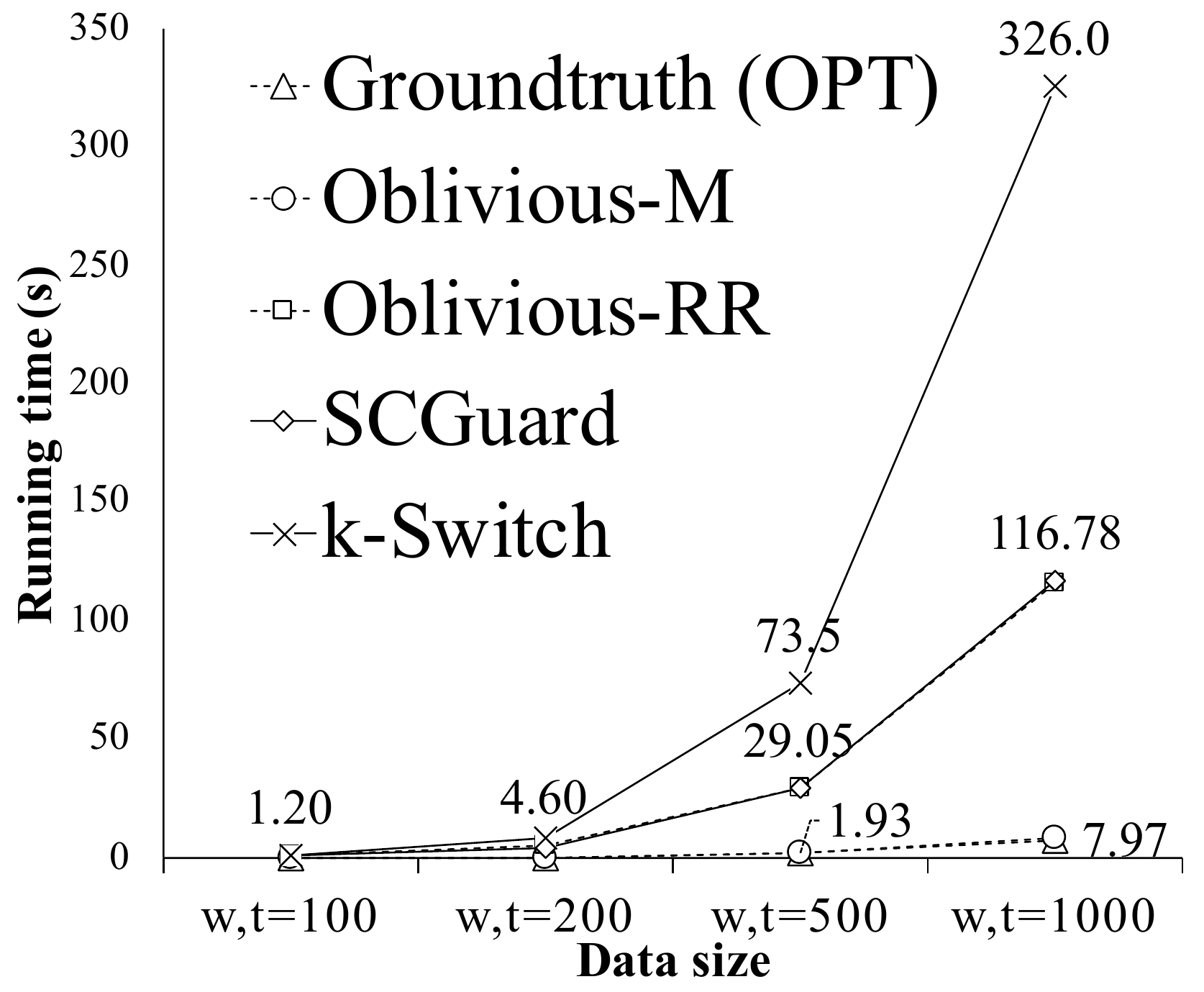}}
		\label{subfig:exp_3_time_syn}}\figureCaptionMargin
    \caption{\small Running time (seconds) on different input sizes.}
    \figureBelowMargin\vspace{0ex}
	\label{fig:exp_3_time}
\end{figure}

\fakeparagraph{Details of results} Fig.~\ref{fig:exp_2_time} shows the running time of different methods (in seconds), over different privacy parameters, on datasets of 100 workers vs. 100 tasks. $k$-Switch obtains around 2 seconds running time on the taxi dataset (Fig.~\ref{subfig:exp_3_time_real}), and strictly less than 2 seconds on the synthetic dataset (Fig.~\ref{subfig:exp_3_time_syn}). The running time is consistent across all privacy levels. When compared with the competing online method SCGuard, $k$-Switch is slightly slower, with a constant factor of 2. 

Varying data size: We look at the asymptotic growth of the running time across different sizes of datasets in Fig.~\ref{fig:exp_3_time}. First, the results verify the quadratic time complexity ($\mathcal{O}(\lambda k^{3} + \lambda n^2\log n)$, see Sec.~\ref{subsec:lambdaOpting}) in terms of $n$, where $n = \max(|W|, |T|)$, the input size. When the input size doubles, from $w=500$ to $w=1000$, the running time increases about 4 times, from 82.4 seconds to 324 seconds on the taxi dataset (Fig.~\ref{subfig:exp_3_time_real}). When compared with SCGuard, SCGuard has the same quadratic time complexity, so it also increases about 4 times, from 30.02 to 120.4 seconds when $w=500$ increases to $w=1000$. $k$-Switch is about a 2 times constant factor slower than SCGuard. The same trend is observed on the synthetic dataset (Fig.~\ref{subfig:exp_3_time_syn}). 

Varying parameter $k$ and $\epsilon$: 
Fig.~\ref{fig:exp_4_time} shows the running time of $k$-Switch, over different $k$, on datasets of different sizes (Fig.~\ref{subfig:exp_4_time_size}) and for different privacy parameters (Fig.~\ref{subfig:exp_4_time_eps}). As expected, when the data size is fixed, the running time of $k$-Switch is larger when $k$ is increased, while the increase is not significant. When we vary the privacy parameter $\epsilon$, the running time of $k$-Switch is stable, across different values of $k$ (Fig.~\ref{subfig:exp_4_time_eps}).

\begin{figure}[t!]\centering \figureTopMargin\vspace{0ex}
	\subfigure[][{\scriptsize On different input sizes}]{
		\scalebox{0.23}[0.23]{\includegraphics{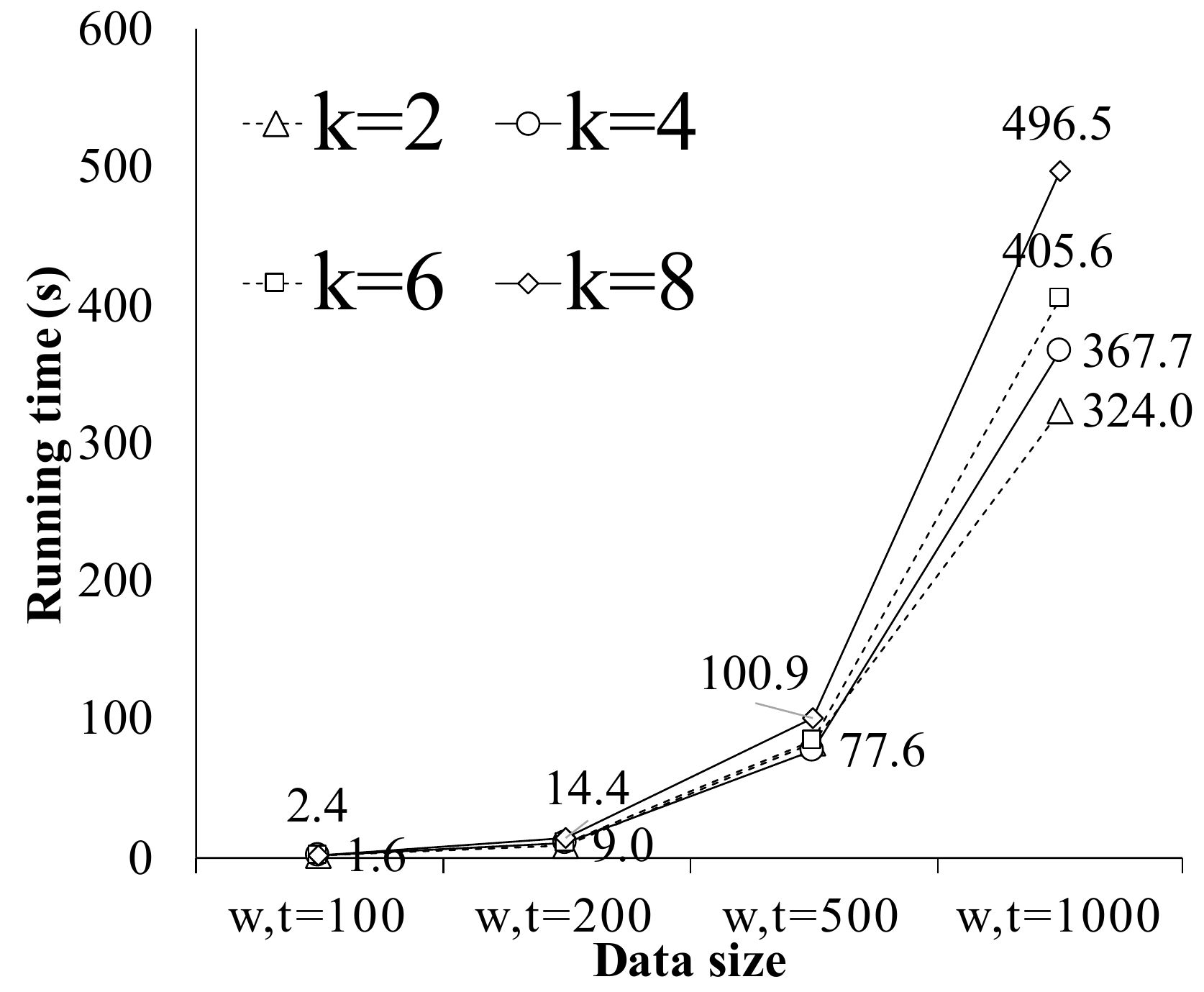}}
		\label{subfig:exp_4_time_size}}
	\subfigure[][{\scriptsize  For different privacy parameters}]{
		\scalebox{0.23}[0.23]{\includegraphics{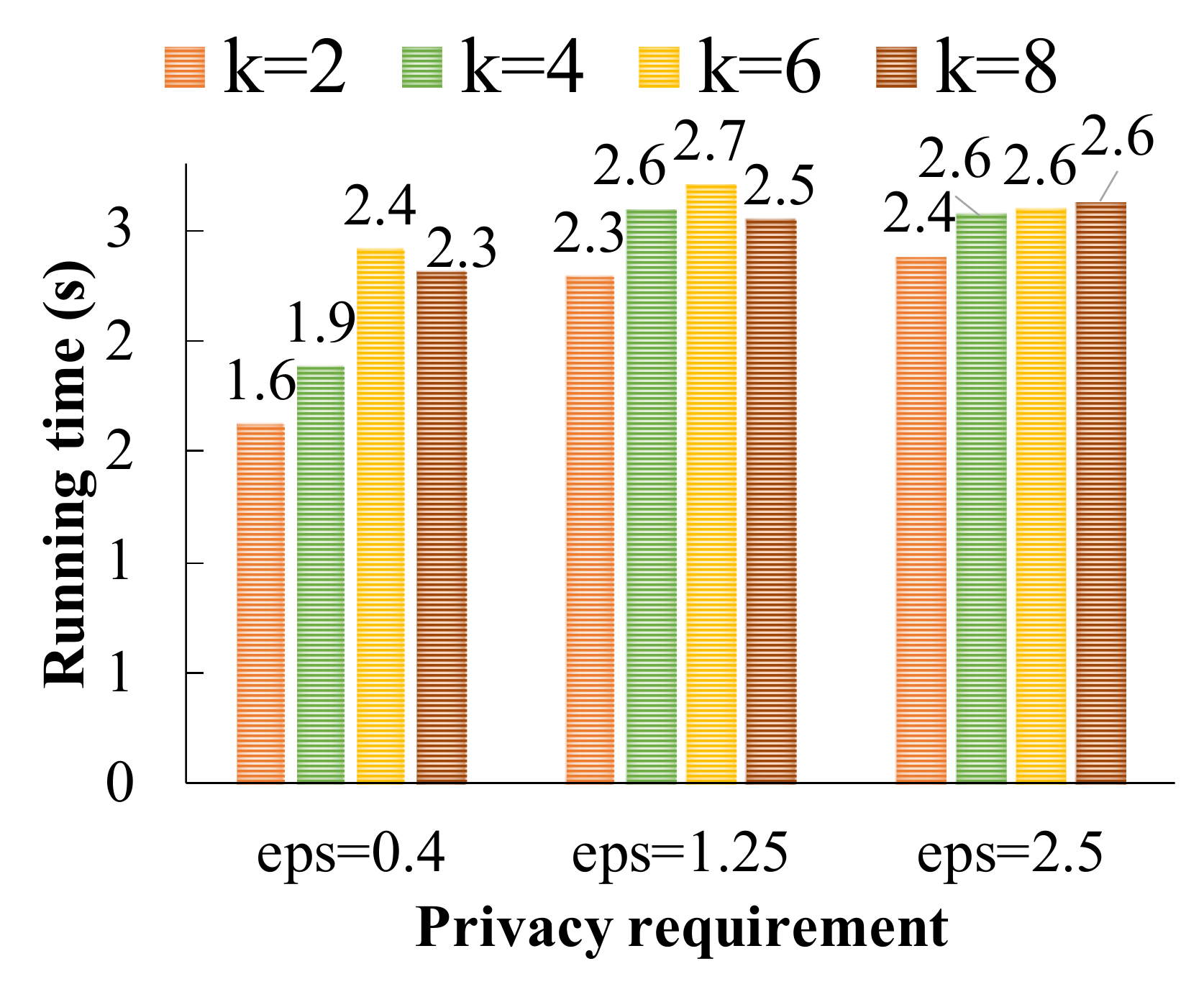}}
		\label{subfig:exp_4_time_eps}}\figureCaptionMargin
    \caption{\small Running time (seconds) of $k$-Switch for different $k$.}
    \figureBelowMargin\vspace{0ex}
	\label{fig:exp_4_time}
\end{figure}

\section{Related Work}
\label{sec:relatedWork}

For privacy-preserving task assignment in spatial crowdsourcing, we have discussed the most directly related online methods \cite{to2018,DBLP:conf/icde/TaoTZSC020} in Sec.~\ref{sec:intro}. Here, we include more related works in a broader context of privacy-preserving spatial crowdsourcing.

\fakeparagraph{Different protection methods} Encryption-based techniques have been used to compute the exact assignment between workers and tasks \cite{DBLP:conf/edbt/LiuCZZZQ17}. The computational cost of such pure encryption-based techniques is high and prohibitive for real-world applications. Other privacy protection technique, such as cloaking, is used to protect locations of workers \cite{DBLP:conf/mdm/PournajafXSG14}, but cloaking is considered as a weaker privacy-preserving technique than Geo-I (Sec.~\ref{subsec:geo-i}), because its assumption on adversaries' prior knowledge.  There are other related works using Geo-I as the privacy standard \cite{DBLP:journals/pvldb/ToGS14, DBLP:journals/tmc/ToGFS17}. However we adopt a stricter privacy model that the server is untrusted from all workers and tasks, and locations of both parties need to be perturbed before released to any other parties. Similar to our work, the batch-based (offline) setting has also been considered in \cite{DBLP:conf/www/WangYHWZM17}. However only the workers are protected. 

\fakeparagraph{Other crowdsourcing setting} There are other related works in the spatial crowdsourcing spectrum \cite{DBLP:journals/tmc/WangPCSWWCQ19, DBLP:journals/www/ZhaiSLLLZZ19, DBLP:journals/tifs/WangZYLHM20}. Different from the task assignment problem, data publishing has been considered in \cite{DBLP:journals/tmc/WangPCSWWCQ19}. The truthful rather than privacy-preserving task assignment is considered in \cite{DBLP:journals/www/ZhaiSLLLZZ19}. Privacy-preserving crowd-sensing is considered in \cite{DBLP:journals/tifs/WangZYLHM20}, and the focus is to protect the locations of workers when they report their sensing results, rather than considering our task assignment setting, where workers need to move to a specified location of the assigned task, and both locations (tasks and workers) are perturbed with differential privacy.  
\vspace{-1ex}
\section{Conclusion}
\label{sec:conclusion}

In this work, we target the Privacy-preserving Batch-based Task Assignment (PBTA) problem, where both workers and tasks use Geo-I to perturb their locations before sending them to the untrusted the server. We propose a novel solution $k$-Switch, which divides the workers into small groups, and uses a secure computation protocol $k$-HE for inner group communication. If workers inside the groups find that switching tasks between them improves the number of successfully assigned tasks, they swap tasks. Extensive experiments demonstrate that $k$-Switch is both effective and efficient, achieving significant utility gains with reasonable system overhead. 
\vspace{-1ex}

\begin{acks}
    Libin Zheng's work is supported by the National Natural Science Foundation of China No. 62102463 and the Basic and Applied basic Research Project of Guangzhou basic Research Program 202102080401. Peng Cheng's work is sponsored by the National Natural Science Foundation of China No. 62102149 and Shanghai Pujiang Program 19PJ1403300. Maocheng, Jiachuan and Lei Chen's work is partially supported by National Key Research and Development Program of China Grant No. 2018AAA0101100, the Hong Kong RGC GRF Project 16209519, CRF Project C6030-18G, C1031-18G, C5026-18G, AOE Project AoE/E-603/18, Theme-based project TRS T41-603/20R, China NSFC No. 61729201, Guangdong Basic and Applied Basic Research Foundation 2019B151530001, Hong Kong ITC ITF grants ITS/044/18FX and ITS/470/18FX, Microsoft Research Asia Collaborative Research Grant, HKUST-NAVER/LINE AI Lab, Didi-HKUST joint research lab, and HKUST-Webank joint research lab grants. Xuemin Lin's work is supported by ARC DP200101338. Corresponding author: Peng Cheng.
\end{acks}

\newpage

\bibliographystyle{ACM-Reference-Format}
\balance
\bibliography{add}


\begin{thebibliography}{29}


\ifx \showCODEN    \undefined \def \showCODEN     #1{\unskip}     \fi
\ifx \showDOI      \undefined \def \showDOI       #1{#1}\fi
\ifx \showISBNx    \undefined \def \showISBNx     #1{\unskip}     \fi
\ifx \showISBNxiii \undefined \def \showISBNxiii  #1{\unskip}     \fi
\ifx \showISSN     \undefined \def \showISSN      #1{\unskip}     \fi
\ifx \showLCCN     \undefined \def \showLCCN      #1{\unskip}     \fi
\ifx \shownote     \undefined \def \shownote      #1{#1}          \fi
\ifx \showarticletitle \undefined \def \showarticletitle #1{#1}   \fi
\ifx \showURL      \undefined \def \showURL       {\relax}        \fi
\providecommand\bibfield[2]{#2}
\providecommand\bibinfo[2]{#2}
\providecommand\natexlab[1]{#1}
\providecommand\showeprint[2][]{arXiv:#2}

\bibitem[\protect\citeauthoryear{??}{fou}{2021}]%
        {foursquare}
 \bibinfo{year}{2021}\natexlab{}.
\newblock \bibinfo{howpublished}{\url{https://foursquare.com}}.
\newblock


\bibitem[\protect\citeauthoryear{??}{gig}{2021}]%
        {gigwalk}
 \bibinfo{year}{2021}\natexlab{}.
\newblock \bibinfo{howpublished}{\url{https://gigwalk.com}}.
\newblock


\bibitem[\protect\citeauthoryear{??}{did}{2021}]%
        {didi}
 \bibinfo{year}{2021}\natexlab{}.
\newblock
  \bibinfo{howpublished}{\url{https://www.didiglobal.com/about-didi/about-us}}.
\newblock


\bibitem[\protect\citeauthoryear{Andreev and R{\"{a}}cke}{Andreev and
  R{\"{a}}cke}{2006}]%
        {DBLP:journals/mst/AndreevR06}
\bibfield{author}{\bibinfo{person}{Konstantin Andreev} {and}
  \bibinfo{person}{Harald R{\"{a}}cke}.} \bibinfo{year}{2006}\natexlab{}.
\newblock \showarticletitle{Balanced Graph Partitioning}.
\newblock \bibinfo{journal}{\emph{Theory Comput. Syst.}} \bibinfo{volume}{39},
  \bibinfo{number}{6} (\bibinfo{year}{2006}), \bibinfo{pages}{929--939}.
\newblock


\bibitem[\protect\citeauthoryear{Andres, Bordenabe, Cjhatzikokolakis, and
  Palamidessi}{Andres et~al\mbox{.}}{2013}]%
        {andres13}
\bibfield{author}{\bibinfo{person}{Miguel~E. Andres},
  \bibinfo{person}{Nicolas~E. Bordenabe}, \bibinfo{person}{Konstatinos
  Cjhatzikokolakis}, {and} \bibinfo{person}{Catuscia Palamidessi}.}
  \bibinfo{year}{2013}\natexlab{}.
\newblock \showarticletitle{Geo-indistinguishability: differential privacy for
  location-based systems}. In \bibinfo{booktitle}{\emph{Proceedings of the 2013
  ACM SIGSAC conference on Computer and communications security}}.
  \bibinfo{pages}{901--914}.
\newblock


\bibitem[\protect\citeauthoryear{Chen, Fu, Zhao, Liu, Xia, Chen, Cheng, Cao,
  and Tong}{Chen et~al\mbox{.}}{2014}]%
        {chen2014gmission}
\bibfield{author}{\bibinfo{person}{Zhao Chen}, \bibinfo{person}{Rui Fu},
  \bibinfo{person}{Ziyuan Zhao}, \bibinfo{person}{Zheng Liu},
  \bibinfo{person}{Leihao Xia}, \bibinfo{person}{Lei Chen},
  \bibinfo{person}{Peng Cheng}, \bibinfo{person}{Caleb~Chen Cao}, {and}
  \bibinfo{person}{Yongxin Tong}.} \bibinfo{year}{2014}\natexlab{}.
\newblock \showarticletitle{gMission: A General Spatial Crowdsourcing
  Platform}.
\newblock \bibinfo{journal}{\emph{Proceedings of the VLDB Endowment}}
  \bibinfo{volume}{7}, \bibinfo{number}{13} (\bibinfo{year}{2014}).
\newblock


\bibitem[\protect\citeauthoryear{Cormen, Leiserson, Rivest, and Stein}{Cormen
  et~al\mbox{.}}{2001}]%
        {cormen01introduction}
\bibfield{author}{\bibinfo{person}{Thomas~H. Cormen},
  \bibinfo{person}{Charles~E. Leiserson}, \bibinfo{person}{Ronald~L. Rivest},
  {and} \bibinfo{person}{Clifford Stein}.} \bibinfo{year}{2001}\natexlab{}.
\newblock \bibinfo{booktitle}{\emph{Introduction to Algorithms}
  (\bibinfo{edition}{2nd} ed.)}.
\newblock \bibinfo{publisher}{The MIT Press}.
\newblock
\showISBNx{0262032937}


\bibitem[\protect\citeauthoryear{Dwork}{Dwork}{2006}]%
        {DBLP:conf/icalp/Dwork06}
\bibfield{author}{\bibinfo{person}{Cynthia Dwork}.}
  \bibinfo{year}{2006}\natexlab{}.
\newblock \showarticletitle{Differential Privacy}. In
  \bibinfo{booktitle}{\emph{33rd {ICALP}}} \emph{(\bibinfo{series}{Lecture
  Notes in Computer Science}, Vol.~\bibinfo{volume}{4052})},
  \bibfield{editor}{\bibinfo{person}{Michele Bugliesi}, \bibinfo{person}{Bart
  Preneel}, \bibinfo{person}{Vladimiro Sassone}, {and} \bibinfo{person}{Ingo
  Wegener}} (Eds.). \bibinfo{publisher}{Springer}, \bibinfo{pages}{1--12}.
\newblock
\urldef\tempurl%
\url{https://doi.org/10.1007/11787006\_1}
\showDOI{\tempurl}


\bibitem[\protect\citeauthoryear{Dwork and Roth}{Dwork and Roth}{2014}]%
        {DBLP:journals/fttcs/DworkR14}
\bibfield{author}{\bibinfo{person}{Cynthia Dwork} {and} \bibinfo{person}{Aaron
  Roth}.} \bibinfo{year}{2014}\natexlab{}.
\newblock \showarticletitle{The Algorithmic Foundations of Differential
  Privacy}.
\newblock \bibinfo{journal}{\emph{Foundations and Trends in Theoretical
  Computer Science}} \bibinfo{volume}{9}, \bibinfo{number}{3-4}
  (\bibinfo{year}{2014}), \bibinfo{pages}{211--407}.
\newblock


\bibitem[\protect\citeauthoryear{Edmonds}{Edmonds}{1965}]%
        {edmond1965}
\bibfield{author}{\bibinfo{person}{Jack Edmonds}.}
  \bibinfo{year}{1965}\natexlab{}.
\newblock \showarticletitle{Paths, trees, and flowers}.
\newblock \bibinfo{journal}{\emph{Canadian Journal of Mathematics}}
  \bibinfo{volume}{17}, \bibinfo{number}{449-467} (\bibinfo{year}{1965}).
\newblock


\bibitem[\protect\citeauthoryear{Isaac and Frenkel}{Isaac and Frenkel}{2018}]%
        {facebook2018}
\bibfield{author}{\bibinfo{person}{Mike Isaac} {and} \bibinfo{person}{Sheera
  Frenkel}.} \bibinfo{year}{2018}\natexlab{}.
\newblock \bibinfo{title}{Facebook Security Breach Exposes Accounts of 50
  Million Users}.
\newblock
  \bibinfo{howpublished}{\url{https://www.nytimes.com/2018/09/28/technology/facebook-hack-data-breach.html}}.
\newblock


\bibitem[\protect\citeauthoryear{Kazemi and Shahabi}{Kazemi and
  Shahabi}{2012}]%
        {kazemi2012}
\bibfield{author}{\bibinfo{person}{Leyla Kazemi} {and} \bibinfo{person}{Cyrus
  Shahabi}.} \bibinfo{year}{2012}\natexlab{}.
\newblock \showarticletitle{GeoCrowd: enabling query answering with spatial
  crowdsourcing}. In \bibinfo{booktitle}{\emph{Proceedings of the 20th
  International Conference on Advances in Geographic Infomation Systems
  (SIGSPATIAL)}}. \bibinfo{pages}{189--198}.
\newblock


\bibitem[\protect\citeauthoryear{Kim, Lu, Constantinou, Shahabi, Wang, and
  Zimmermann}{Kim et~al\mbox{.}}{2014}]%
        {kim2014mediaq}
\bibfield{author}{\bibinfo{person}{Seon~Ho Kim}, \bibinfo{person}{Ying Lu},
  \bibinfo{person}{Giorgos Constantinou}, \bibinfo{person}{Cyrus Shahabi},
  \bibinfo{person}{Guanfeng Wang}, {and} \bibinfo{person}{Roger Zimmermann}.}
  \bibinfo{year}{2014}\natexlab{}.
\newblock \showarticletitle{Mediaq: mobile multimedia management system}. In
  \bibinfo{booktitle}{\emph{Proceedings of the 5th ACM Multimedia Systems
  Conference}}. ACM, \bibinfo{pages}{224--235}.
\newblock


\bibitem[\protect\citeauthoryear{Liu, Chen, Zhu, Zhang, Zhang, and Qiu}{Liu
  et~al\mbox{.}}{2017}]%
        {DBLP:conf/edbt/LiuCZZZQ17}
\bibfield{author}{\bibinfo{person}{Bozhong Liu}, \bibinfo{person}{Ling Chen},
  \bibinfo{person}{Xingquan Zhu}, \bibinfo{person}{Ying Zhang},
  \bibinfo{person}{Chengqi Zhang}, {and} \bibinfo{person}{Weidong Qiu}.}
  \bibinfo{year}{2017}\natexlab{}.
\newblock \showarticletitle{Protecting Location Privacy in Spatial
  Crowdsourcing using Encrypted Data}. In \bibinfo{booktitle}{\emph{{EDBT}}}.
  \bibinfo{publisher}{OpenProceedings.org}, \bibinfo{pages}{478--481}.
\newblock


\bibitem[\protect\citeauthoryear{Paillier}{Paillier}{1999}]%
        {DBLP:conf/eurocrypt/Paillier99}
\bibfield{author}{\bibinfo{person}{Pascal Paillier}.}
  \bibinfo{year}{1999}\natexlab{}.
\newblock \showarticletitle{Public-Key Cryptosystems Based on Composite Degree
  Residuosity Classes}. In \bibinfo{booktitle}{\emph{{EUROCRYPT}}}
  \emph{(\bibinfo{series}{Lecture Notes in Computer Science},
  Vol.~\bibinfo{volume}{1592})}. \bibinfo{publisher}{Springer},
  \bibinfo{pages}{223--238}.
\newblock


\bibitem[\protect\citeauthoryear{Pournajaf, Xiong, Sunderam, and
  Goryczka}{Pournajaf et~al\mbox{.}}{2014}]%
        {DBLP:conf/mdm/PournajafXSG14}
\bibfield{author}{\bibinfo{person}{Layla Pournajaf}, \bibinfo{person}{Li
  Xiong}, \bibinfo{person}{Vaidy~S. Sunderam}, {and} \bibinfo{person}{Slawomir
  Goryczka}.} \bibinfo{year}{2014}\natexlab{}.
\newblock \showarticletitle{Spatial Task Assignment for Crowd Sensing with
  Cloaked Locations}. In \bibinfo{booktitle}{\emph{{MDM} {(1)}}}.
  \bibinfo{publisher}{{IEEE} Computer Society}, \bibinfo{pages}{73--82}.
\newblock


\bibitem[\protect\citeauthoryear{Snyder}{Snyder}{1997}]%
        {snyder1997flattening}
\bibfield{author}{\bibinfo{person}{John~P Snyder}.}
  \bibinfo{year}{1997}\natexlab{}.
\newblock \bibinfo{booktitle}{\emph{Flattening the earth: two thousand years of
  map projections}}.
\newblock \bibinfo{publisher}{University of Chicago Press}.
\newblock


\bibitem[\protect\citeauthoryear{Tao, Tong, Zhou, Shi, Chen, and Xu}{Tao
  et~al\mbox{.}}{2020}]%
        {DBLP:conf/icde/TaoTZSC020}
\bibfield{author}{\bibinfo{person}{Qian Tao}, \bibinfo{person}{Yongxin Tong},
  \bibinfo{person}{Zimu Zhou}, \bibinfo{person}{Yexuan Shi},
  \bibinfo{person}{Lei Chen}, {and} \bibinfo{person}{Ke Xu}.}
  \bibinfo{year}{2020}\natexlab{}.
\newblock \showarticletitle{Differentially Private Online Task Assignment in
  Spatial Crowdsourcing: {A} Tree-based Approach}. In
  \bibinfo{booktitle}{\emph{{ICDE}}}. \bibinfo{publisher}{{IEEE}},
  \bibinfo{pages}{517--528}.
\newblock


\bibitem[\protect\citeauthoryear{To, Ghinita, Fan, and Shahabi}{To
  et~al\mbox{.}}{2017}]%
        {DBLP:journals/tmc/ToGFS17}
\bibfield{author}{\bibinfo{person}{Hien To}, \bibinfo{person}{Gabriel Ghinita},
  \bibinfo{person}{Liyue Fan}, {and} \bibinfo{person}{Cyrus Shahabi}.}
  \bibinfo{year}{2017}\natexlab{}.
\newblock \showarticletitle{Differentially Private Location Protection for
  Worker Datasets in Spatial Crowdsourcing}.
\newblock \bibinfo{journal}{\emph{{IEEE} Trans. Mob. Comput.}}
  \bibinfo{volume}{16}, \bibinfo{number}{4} (\bibinfo{year}{2017}),
  \bibinfo{pages}{934--949}.
\newblock


\bibitem[\protect\citeauthoryear{To, Ghinita, and Shahabi}{To
  et~al\mbox{.}}{2014}]%
        {DBLP:journals/pvldb/ToGS14}
\bibfield{author}{\bibinfo{person}{Hien To}, \bibinfo{person}{Gabriel Ghinita},
  {and} \bibinfo{person}{Cyrus Shahabi}.} \bibinfo{year}{2014}\natexlab{}.
\newblock \showarticletitle{A Framework for Protecting Worker Location Privacy
  in Spatial Crowdsourcing}.
\newblock \bibinfo{journal}{\emph{Proc. {VLDB} Endow.}} \bibinfo{volume}{7},
  \bibinfo{number}{10} (\bibinfo{year}{2014}), \bibinfo{pages}{919--930}.
\newblock


\bibitem[\protect\citeauthoryear{To, Shahabi, and Xiong}{To
  et~al\mbox{.}}{2018}]%
        {to2018}
\bibfield{author}{\bibinfo{person}{Hien To}, \bibinfo{person}{Cyrus Shahabi},
  {and} \bibinfo{person}{Li Xiong}.} \bibinfo{year}{2018}\natexlab{}.
\newblock \showarticletitle{Privacy-Preserving Online Task Assignment in
  Spatial Crowdsourcing with Untrusted Server}. In
  \bibinfo{booktitle}{\emph{Proceedings of the 34th IEEE International
  Conference on Data Engineering (ICDE)}}. \bibinfo{pages}{833--844}.
\newblock


\bibitem[\protect\citeauthoryear{Tong, Zhou, Zeng, Chen, and Shahabi}{Tong
  et~al\mbox{.}}{2019}]%
        {DBLP:journals/vldb/yongxin19}
\bibfield{author}{\bibinfo{person}{Yongxin Tong}, \bibinfo{person}{Zimu Zhou},
  \bibinfo{person}{Yuxiang Zeng}, \bibinfo{person}{Lei Chen}, {and}
  \bibinfo{person}{Cyrus Shahabi}.} \bibinfo{year}{2019}\natexlab{}.
\newblock \showarticletitle{Spatial crowdsourcing: a survey}.
\newblock \bibinfo{journal}{\emph{The VLDB Journal}} (\bibinfo{year}{2019}).
\newblock
\showISBNx{0949-877X}
\urldef\tempurl%
\url{https://doi.org/10.1007/s00778-019-00568-7}
\showDOI{\tempurl}


\bibitem[\protect\citeauthoryear{Wang, Yang, Han, Wang, Zhang, and Ma}{Wang
  et~al\mbox{.}}{2017}]%
        {DBLP:conf/www/WangYHWZM17}
\bibfield{author}{\bibinfo{person}{Leye Wang}, \bibinfo{person}{Dingqi Yang},
  \bibinfo{person}{Xiao Han}, \bibinfo{person}{Tianben Wang},
  \bibinfo{person}{Daqing Zhang}, {and} \bibinfo{person}{Xiaojuan Ma}.}
  \bibinfo{year}{2017}\natexlab{}.
\newblock \showarticletitle{Location Privacy-Preserving Task Allocation for
  Mobile Crowdsensing with Differential Geo-Obfuscation}. In
  \bibinfo{booktitle}{\emph{{WWW}}}. \bibinfo{publisher}{{ACM}},
  \bibinfo{pages}{627--636}.
\newblock


\bibitem[\protect\citeauthoryear{Wang, Zhang, Yang, Lim, Han, and Ma}{Wang
  et~al\mbox{.}}{2020}]%
        {DBLP:journals/tifs/WangZYLHM20}
\bibfield{author}{\bibinfo{person}{Leye Wang}, \bibinfo{person}{Daqing Zhang},
  \bibinfo{person}{Dingqi Yang}, \bibinfo{person}{Brian~Y. Lim},
  \bibinfo{person}{Xiao Han}, {and} \bibinfo{person}{Xiaojuan Ma}.}
  \bibinfo{year}{2020}\natexlab{}.
\newblock \showarticletitle{Sparse Mobile Crowdsensing With Differential and
  Distortion Location Privacy}.
\newblock \bibinfo{journal}{\emph{{IEEE} Trans. Inf. Forensics Secur.}}
  \bibinfo{volume}{15} (\bibinfo{year}{2020}), \bibinfo{pages}{2735--2749}.
\newblock


\bibitem[\protect\citeauthoryear{Wang, Pang, Chen, Shao, Wang, Wu, Chen, and
  Qi}{Wang et~al\mbox{.}}{2019}]%
        {DBLP:journals/tmc/WangPCSWWCQ19}
\bibfield{author}{\bibinfo{person}{Zhibo Wang}, \bibinfo{person}{Xiaoyi Pang},
  \bibinfo{person}{Yahong Chen}, \bibinfo{person}{Huajie Shao},
  \bibinfo{person}{Qian Wang}, \bibinfo{person}{Libing Wu},
  \bibinfo{person}{Honglong Chen}, {and} \bibinfo{person}{Hairong Qi}.}
  \bibinfo{year}{2019}\natexlab{}.
\newblock \showarticletitle{Privacy-Preserving Crowd-Sourced Statistical Data
  Publishing with An Untrusted Server}.
\newblock \bibinfo{journal}{\emph{{IEEE} Trans. Mob. Comput.}}
  \bibinfo{volume}{18}, \bibinfo{number}{6} (\bibinfo{year}{2019}),
  \bibinfo{pages}{1356--1367}.
\newblock


\bibitem[\protect\citeauthoryear{Xu, Li, Guan, Zhang, Li, Nan, Liu, Bian, and
  Ye}{Xu et~al\mbox{.}}{2018}]%
        {didi2018}
\bibfield{author}{\bibinfo{person}{Zhe Xu}, \bibinfo{person}{Zhixin Li},
  \bibinfo{person}{Qingwen Guan}, \bibinfo{person}{Dingshui Zhang},
  \bibinfo{person}{Qiang Li}, \bibinfo{person}{Junxiao Nan},
  \bibinfo{person}{Chunyang Liu}, \bibinfo{person}{Wei Bian}, {and}
  \bibinfo{person}{Jieping Ye}.} \bibinfo{year}{2018}\natexlab{}.
\newblock \showarticletitle{Large-Scale Order Dispatch in On-Demand
  Ride-Hailing Platforms: A Learning and Planning Approach}. In
  \bibinfo{booktitle}{\emph{KDD '18 Proceedings of the 24th ACM SIGKDD
  International Conference on Knowledge Discovery and Data Mining}}.
  \bibinfo{pages}{905--913}.
\newblock


\bibitem[\protect\citeauthoryear{Zeng, Tong, Chen, and Zhou}{Zeng
  et~al\mbox{.}}{2018}]%
        {DBLP:conf/icde/ZengTCZ18}
\bibfield{author}{\bibinfo{person}{Yuxiang Zeng}, \bibinfo{person}{Yongxin
  Tong}, \bibinfo{person}{Lei Chen}, {and} \bibinfo{person}{Zimu Zhou}.}
  \bibinfo{year}{2018}\natexlab{}.
\newblock \showarticletitle{Latency-Oriented Task Completion via Spatial
  Crowdsourcing}. In \bibinfo{booktitle}{\emph{{ICDE}}}.
  \bibinfo{pages}{317--328}.
\newblock


\bibitem[\protect\citeauthoryear{Zhai, Sun, Liu, Li, Liu, Zhao, and Zheng}{Zhai
  et~al\mbox{.}}{2019}]%
        {DBLP:journals/www/ZhaiSLLLZZ19}
\bibfield{author}{\bibinfo{person}{Dongjun Zhai}, \bibinfo{person}{Yue Sun},
  \bibinfo{person}{An Liu}, \bibinfo{person}{Zhixu Li},
  \bibinfo{person}{Guanfeng Liu}, \bibinfo{person}{Lei Zhao}, {and}
  \bibinfo{person}{Kai Zheng}.} \bibinfo{year}{2019}\natexlab{}.
\newblock \showarticletitle{Towards secure and truthful task assignment in
  spatial crowdsourcing}.
\newblock \bibinfo{journal}{\emph{World Wide Web}} \bibinfo{volume}{22},
  \bibinfo{number}{5} (\bibinfo{year}{2019}), \bibinfo{pages}{2017--2040}.
\newblock


\bibitem[\protect\citeauthoryear{Zheng, Chen, and Ye}{Zheng
  et~al\mbox{.}}{2018}]%
        {DBLP:journals/pvldb/Zheng0Y18}
\bibfield{author}{\bibinfo{person}{Libin Zheng}, \bibinfo{person}{Lei Chen},
  {and} \bibinfo{person}{Jieping Ye}.} \bibinfo{year}{2018}\natexlab{}.
\newblock \showarticletitle{Order Dispatch in Price-aware Ridesharing}.
\newblock \bibinfo{journal}{\emph{Proc. {VLDB} Endow.}} \bibinfo{volume}{11},
  \bibinfo{number}{8} (\bibinfo{year}{2018}), \bibinfo{pages}{853--865}.
\newblock


\end{thebibliography}

\newpage
\appendix
\section{Technical details}

\subsection{Oblivious-RR}
\label{subsec:appendix_orr}

First, for the fractional reachability graph, the weight for an edge $w_i, t_j$ is given by $\Pr(d(w_i, t_j) \leq R_w | d(w'_i, t'_j))$, i.e., the likelihood that the distance between worker $w_i$ and task $t_j$ is indeed smaller than the range $R_w$ using true locations, given the observed distance between their perturbed locations. We directly adopt analytical approach proposed in \cite{to2018} (see Sec. IV-B of \cite{to2018} for details), which uses binomial distribution to approximate the planer Laplace distribution in Geo-I. The calculation is approximate not exact, however it suffices as we only use it to measure the relative high/low likelihood of reachability. We then run the max-flow algorithms to obtain the maximum flow. Note the difference of this step with Figure \ref{fig:example_3}, which obtains the maximum \textit{cardinality} because of the integral 0 or 1 capacity of edges, here we obtain the flow with the maximum-weight flow. The obtained flow is shown in Fig.~\ref{subfig:orr_2_left}.

The last step is to randomly round the fractional flow to obtain a matching. For each worker $w_i$, we randomly select a task $t_j$ with probability proportional to the units of flow sent on the edge $(w_i, t_j)$. Using Fig.~\ref{subfig:orr_2_left} as an example, for $w_1$, we set the probability of selecting $t_1$ as $0.1 / (0.1 + 0.4 + 0.48) = 0.1 / 0.98 \approx 0.102$, selecting $t_2$ as $0.4 / 0.98 \approx 0.408 $ and selecting $t_3$ as $0.48 / 0.98 \approx 0.49$. $t_3$ has the highest probability of being selected. 

In Fig.~\ref{subfig:orr_2_right}, we show a particular matching obtained after the randomized rounding. If we use the true locations shown in Fig.~\ref{fig:example_4} to check the reachability constraint, it turns out that this matching has 3 valid assigned tasks, which equals to the optimal matching.

The details of Oblivious-RR is shown in Algorithm~\ref{algo:o_rr}. Similar to Oblivious-M, it constructs the flow network by adding the superficial source/sink nodes. Different from Oblivious-M, the edge weight is calculated based on the reachability likelihood at Line 10 (see Sec. IV-B of \cite{to2018} for how to calculate the likelihood). Line 15-20 execute the randomized rounding, selecting a task for a worker randomly proportional to the units of flow sent on the edge connecting them. The time complexity of Oblivious-RR is the same with Oblivious-M, as it only adds a post randomized rounding which takes $\mathcal{O}(n^2)$ time, with $n = \max(|W|, |T|)$. Overall, Oblivious-RR runs in $\mathcal{O}(n^3)$ time.

\begin{algorithm}[t]
	\DontPrintSemicolon
	\KwIn{Worker set $W$ and task set $T$. Perturbed locations $l_{w'}$ and $l_{t'}$ for each $w\in W$ and $t \in T$. }
	\KwOut{An assignment (matching) $M$ between workers and tasks.}

	$g:=$ an empty flow network \;
	$s:=$ source node; $d:=$ target (sink) node \;
	$g$.addNode($s$); $g$.addNode($d$) \; 
	\ForEach{$w \in W$}{
		$e := (s, w, capacity=1.0)$; $g$.addEdge($e$)\;
		
	}

	\ForEach{$t \in T$}{
		$e := (t, d, capacity=1.0)$; $g$.addEdge($e$)	\;
	}

	\ForEach{$w \in W$}{
		\ForEach{$t \in T$}{
			$prob := $ ComputeProb($l_{w'}, l_{t'}, R_w$) \;
			$e := (s, w, capacity=prob)$ \;
			$g$.addEdge($e$) \;
			
		}
	}

	$f := $ Ford-Fulkerson($g$) \;
	$M:= \{\}$ \;
	\ForEach{$w \in W$}{
		$wt_{sum}:= 0$ \;
		\ForEach{$t \in T$}{
			$wt[t] := f(w, t)$ \;
			$wt_{sum} += wt[t]$ \;
		}
		$t_{selected} :=$ sample $t \in T$ with a prob. $wt[t] / wt_{sum}$ \;
		$M$.insert($w, t_{selected}$) \;
	}

	\Return{$M$}\;
	\caption{\texttt{Oblivious-RR} }
	\label{algo:o_rr}
\end{algorithm}

\subsection{k-Switch}
\label{subsec:appendix_kswitch}

\noindent \textbf{Extended proof of Theorem~\ref{theo:kpgk2}}: 
\begin{proof}
    (Continued.) Recall that a matching on a graph is a set of edges without common vertices. The maximum weight matching is a matching with maximum sum of the weight on the edges of the matching. 

	We show two directions. First, if we have a $2$-division $D = \{g_1, \ldots, g_d\}$ of the KGP, it could be mapped to a matching $m$ on $M_\mathcal{G}$. The matching $m$ is obtained by only including edges within each $2$-group $g_i \in D$, i.e., $m = \{e=(p_i, p_j) | \exists q \to p_i, p_j \in g_q\}$. We show the edge set $m$ is indeed a matching. According to Def.~\ref{def:k-division}, $2$-groups are non-overlapping. This ensure that no edges in $m$ shares common vertices, as they only include edges within each $2$-groups, and $2$-groups are non-overlapping. So edges $m$ do not share vertices, and thus form a matching. 

	Then, we show the other direction, if we have a matching $m = \{e=(p_i, p_j)\}$ on $M_\mathcal{G}$, it could be mapped to a $k$-division $D$ of KGP. We create the $k$-division as follows: for each edge $e=(p_i, p_j)$, we create a $2$-group $g=(p_i, p_j)$ and insert it to $D$. As $m$ is matching, so edges don't share common vertices, all created $2$-groups don't share any matched-pairs, and the collection $D$ forms a $k$-division. 
\end{proof}

\noindent \textbf{Graph transformation}: 
\begin{figure}[t!]\centering
	\scalebox{0.35}[0.35]{\includegraphics{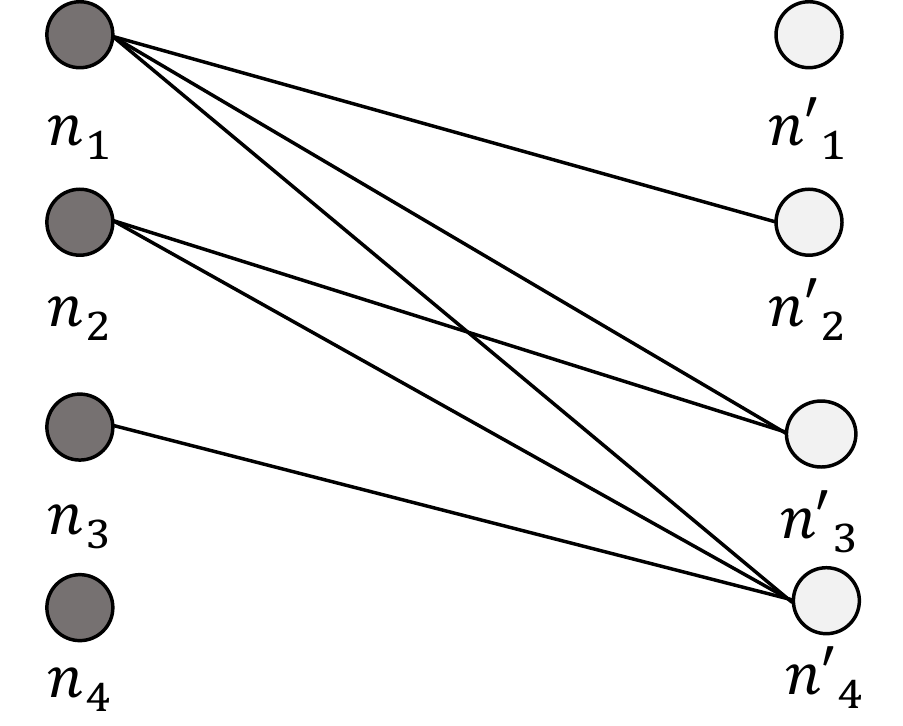}}\vspace{-1ex}
	\caption{\small An illustrative example of the graph transformation in Alg.~\ref{algo:graphTransform}.}
	\figureBelowMargin
	\label{fig:graphTransformation}
\end{figure}

\begin{algorithm}[t]
	\DontPrintSemicolon
	\KwIn{Complete graph $G=(V, E)$ }
	\KwOut{Bipartite graph $G'$}
	Initialize an empty graph $G'$ \\
	\ForEach{node $n_i \in V$}{
		Insert $n_i$ to $G'$ \\
		Create a new node $n'_{i}$, insert it to $G'$ \\
	}
	\ForEach{node $n_i \in G$}{
		\ForEach{node $n_j \in G$}{
			$e:=$ edge between $n_i$ and $n_j$ \\
			$w:= w(e)$ \\
			Add edge $e'=(n_i, n'_{j}, weight=w)$ to $G'$\\ 
		}
	}
	\Return{$G'$}
	\caption{GraphTransform}
	\label{algo:graphTransform}
\end{algorithm}

Alg.~\ref{algo:graphTransform} transforms a complete graph to a bipartite graph. See Figure \ref{fig:graphTransformation} for an example of a graph containing 4 nodes. The basic idea is to copy $v$ nodes to $2v$ nodes, where the 1st node is copied to its copied node (denoted by $n'_i$). The original and the copied nodes are put in sets $\mathcal{L}$ and $\mathcal{R}$, respectively. Then we connect every node in set $\mathcal{L}$ with every other node in set $\mathcal{R}$, except for the node copy of itself. Also, we don't add redundant edges, meaning we only add edges between $n_2$ to node $n'_3$ and $n'_4$, but not to $n'_1$, because edge $(n_2, n'_1)$ is the same as $(n_1, n'_2)$. \\

\noindent \textbf{Extended proof of Theorem~\ref{theo:kpgk3}}: 
\begin{proof}
    We show that there is a \textit{yes} instance for the PBGP iff. there is a \textit{yes} instance for KGP. We first define the decision version of KGP: given a positive integer $J$, does there exist a $k$-division $D = \{g_1, \ldots, g_d\}$, such that $\text{score}(D) \leq J$?
	
	For KGP, we construct a graph $G' = (V', E')$ similar to the proof of Theorem~\ref{theo:kpgk2} when $k=2$. The vertex set $V'$ corresponds to each matched worker-task pair $p$ from the baseline matching $M_0$. Between every vertex $p_i$ and $p_j$ in $G'$, we add an edge $e=(p_i, p_j)$. This edge corresponds to a $2$-group, with edge weight $w(e)=\text{OScore}((p_i, p_j))$. 
	
	If-direction: if there is a \textit{yes} instance to the KGP, we could find a \textit{yes} instance to the PBGP. Note that for the rest of the proof, we assume the input for the KGP, the baseline matching $M_0$ has a size $|M_0|$ being a multiple of $k$, for simplicity. The hardness result holds w.l.o.g. because the following proof holds on the special case of KGP, so the general version is even harder. Say for a given integer $J$, we could obtain a $k$-division $D = \{g_1, \ldots, g_d\}$, such that $\text{score}(D) \leq J$. Such $k$-division corresponds to a $p$-partition. For each $k$-group, it contains $k$ nodes. We have $p=\frac{|V'|}{k}$ $k$-groups in total. The score of the $k$-division is $\text{score}(D) \leq J$, and we look closely at what edges are included. Since $\text{score}(D) = \sum_{i} \text{OScore}(g_i)$, it sums up all OScore of all the $k$-groups. For each $k$-group, $\text{OScore}(g_i) = \sum_{1 \leq s,t \leq k}\text{OScore}(\{p_{s}, p_{t}\})$, summing up the edges weight between the nodes within a particular $k$-group. This $k$-division corresponds to a $p$-partition, and we let $W = J$. The weight sum  corresponds to all the edges which are within each partition, which is the  $k$-groups. And because $\text{score}(D) \leq J$, we know the the sum of weights within each partition in the $p$-partition is smaller or equal than $W$. 
	
	Only-if direction: if there is a \textit{yes} instance to the PBGP problem, then we could obtain a \textit{yes} instance for the KGP. Given the graph and the $p$-partition, we know we have $V=V_1 \cup \cdots \cup V_p$, $|V_1| = \cdots |V_p| = \frac{|V|}{p}$. And for a given positive integer $W$, for the edges $E' \subset E$ that have two endpoints in the two different sets $V_i$, $E' = \{(v_i, v_j)\in E | v_i \in V_i, v_j \in V_j , i \neq j \}$, the sum of the weights on all such edges is smaller or equal to the given integer $W$, $\sum_{e \in E'}w(e) \leq W$. First we transform the graph to be a complete graph by adding edges between nodes that don't edges, and setting the edge weight to be 0. Since we know we have a \textit{yes} instance, adding such 0 weight edges to the graph would not increase the cut edges (edges cross different partitions), and the transformed instance would still be a \textit{yes} instance. We then transform each $V_i$ to a corresponding $k$-group $g_i$, where we set $k=\frac{|V|}{p}$. Each $g_i$ contains all the nodes in $V_i$. Since we've transformed the graph to be complete graph, each pair of nodes in $g_i$ has edges between them, and for any edge $e = (n_i, n_j)$ we define $\text{OScore}(\{n_i, n_j\}) = w(e)$. So obviously for this constructed $k$-size $g_i$, we have the $\text{OScore}$ defined, by summing up the weight of all the edges within the group. Finally, now we have a $k$-division $D = \{g_i\}$, where each $g_i$ is transformed from the partition $V_i$, and has size $k$. We define the score for $G$ to be $\text{score}(D) = \sum_{i} \text{OScore}(g_i)$. Because we know all the cross-partition edges sum, $W'$ is less or equal to $W$, so all the in-partition edges sum, which is $\text{score}(D) =  W' \leq W$. Set $J = W$, we've obtained a \textit{yes} instance to KGP, $\text{score}(D) \leq J$. 
\end{proof}

\noindent \textbf{Details of Greedy-Grouping}: 

\begin{algorithm}[t]
	\DontPrintSemicolon
	\KwIn{ A baseline matching $M_0$. Perturbed locations $l_{w'}$ and $l_{t'}$ for each $w\in W$ and $t \in T$. }
	\KwOut{A $k$-division $D$.}

	Initialize a heap $h$ \;
	\ForEach{$p_1 \in M_0$}{
		$used[p_1]=$ False \;
		\ForEach{$p_2 \in M_0$}{
			\If{$p_1.w \neq p_2.w$}{
				oscore$:= d(l'_{p_1.w}, l'_{p_2.w}) + d(l'_{p_1.t}, l'_{p_2.t})$ \;
				$h$.insert($\{p_1, p_2\}$, oscore) \tcp*{Push $2$-group $\{p_1, p_2\}$ with its OScore into the heap}
			}
		}
	}

	$d := \lceil |M_0| / k \rceil$ \tcp*{Calculate how many $k$-groups to create} 
	$D :=  \{\}$ \;
	\For{$i:=0$; $i < d$; $i=i+1$}{
		$g:=\{\}$ \;
		\While{$g.size() < k$}{
			\If{$size == k-1$}{
				\While{True}{
					$p:=$ random($|M_0|$) \;
					\If{$used[p.w]$ is False}{
						break \;
					}
				}
				$used[p]=$ True\;
				$g$.Insert($p$) \;
				break \;
			}
			\While{True}{
				$p_1, p_2:=h$.pop() \;
				\If{$used[p_1.w]$ is False and $used[p_2.w]$ is False }{
					break \;
				}	
			}
			$used[p_1]=used[p_2]=$ True\;
			$g$.Insert($p_1$) \;
			$g$.Insert($p_2$) \;
		}
		$D$.Insert($g$) \;
	}
	\Return{$D$}\;
	\caption{\texttt{Greedy-Grouping} }
	\label{algo:greedy_grouping_full}
\end{algorithm}

Alg.~\ref{algo:greedy_grouping_full} is the full detailed algorithm. At Line~\ref{algo:line:heap}, the algorithm uses a heap storing all combinations of pairs in the baseline matching $\{p_i, p_j\}$ with its OScore as the sorted key. In this way, every time we could pop the pair $\{p_i, p_j\}$ with the smallest OScore, and if both of them have not been inserted to any $k$-group, we add them to the current $k$-group. After each $k$-group is formed, we continue to the next one until a $k$-division is obtained and returned. Line 10-28 greedily add two matched pairs to the current $k$-group $g$, until its size reaches $k$. In total, we form $d$ groups, as calculated at Line 8. The time complexity of Alg.~ \ref{algo:greedy_grouping} is dominated by the heap construction, which takes $\mathcal{O}(e\log e)$, where $e$ is the total number of elements inserted to the heap. We know $e$ is all the combinations of matched pairs, as shown at Line 2 and 4, so $e=n^2$, where $n=\max(|W|,|T|)$, the size of the baseline matching. In conclusion, for Alg.~\ref{algo:greedy_grouping_full} has an $\mathcal{O}(n^2\log n^2) \to \mathcal{O}(n^2\log n)$ time complexity. 

\subsection{Details of k-HE protocol}
\label{subsec:appendix_kHE}

The detailed $k$-HE protocol is shown in Alg.~\ref{algo:kHE_full}. 

\begin{algorithm}[t]
	\DontPrintSemicolon
	\KwIn{A $k$-group $g = \{p_{1}, \cdots , p_{k} \}$ }
	\KwOut{The set of modified worker-task matching $M_{g}$ for $g$}
	Let $w_1, \cdots, w_k$ be the workers inside $k$-group $g$\\
	Let $t_1, \cdots, t_k$ be the tasks inside $g$\\
	Randomly elect two nodes (task or worker) as proxy servers $P_a$ and $P_b$ \\
	$P_b$ generates its public key $pk$ and private key $sk$ \\
	$P_b$ sends its public key to $P_a$ and all workers and tasks \\
	\ForEach{$w$}{
		Worker $w$ sends their encrypted  location $(E_{pk}(l_{w}.x), E_{pk}(l_{w}.y))$ to proxy $P_a$ \\
	}
	\ForEach{$t$}{
		Task $t$ sends their encrypted location $(E_{pk}(l_{t}.x), E_{pk}(l_{t}.y))$ to proxy $P_a$\\
	}
	\ForEach{$w$}{
		\ForEach{$t$}{
			$P_a$ and $P_b$ use \texttt{SecureDistanceCalculation} (Algorithm \ref{algo:secureDistanceCalc}) to calculate $E_{pk}(m)=E_{pk}(|l_{w}- l_{t}|^{2})$. $m$ is sent to $P_b$ \\
			$P_b$ decrypts message using its private key and store the true distance $m'=D_{sk}(m) = |l_{w}- l_{t}|^{2}$ \\
		}
	}
	$P_b$ runs an exact matching algorithm to obtain a new matching $M_{g}$ \;
	
	\Return{$M_{g}$}\;
	\caption{$k$-HE protocol}
	\label{algo:kHE_full}
\end{algorithm}

Step 1. two random nodes in the $k$-group are elected to serve as the proxy servers $P_a$ and $P_b$. Note that the nodes could be either a worker or a task. The public keys $pk$ are generated for $P_b$, and sent to all parties. 

Step 2. all the tasks and workers send the encrypted true location to $P_a$. The true location for a participant $v$ ($v$ is a worker $w$ or a task $t$) is represented by a 2-dimensional tuple $(l_{v}.x, l_{v}.y)$, and the encrypted coordinates are $(E_{pk}(l_{v}.x), E_{pk}(l_{v}.y))$. 

Step 3. the two proxy servers $P_a$ and $P_b$ utilize an interactive communication scheme to compute the true distance between the workers and tasks within the group. The interactive communication scheme is shown in Alg.~\ref{algo:secureDistanceCalc}. Most of the steps follow Algorithm 1 in HESI \cite{DBLP:conf/edbt/LiuCZZZQ17}. The true distance is obtained and stored at $P_b$. The purpose of having a two-server structure is to compute the multiplication of two secret numbers, which needs the plaintext of one of the operands. All the communication between the two proxy servers are encrypted using the public key $pk$. In the end, $p_b$ uses its private key to decrypt the true distance between workers/tasks in the group.

\begin{algorithm}[t]
	\DontPrintSemicolon
	\KwIn{Encrypted location of a task $E_pk(l_{t})$, true location of a worker $l_{w}$}
	\KwOut{The true distance $|l_w-l_t|$}
	$P_a$ computes $E_{pk}(l_w.x-l_t.x)=E_{pk}(l_w.x)*E_{pk}(l_t.x)^{N-1}$ \\
	$P_a$ computes $E_{pk}(l_w.y-l_t.y)=E_{pk}(l_w.y)*E_{pk}(l_t.y)^{N-1}$ \\
	$P_a$ and $P_b$ uses the \textit{SecMul} protocol \protect\cite{DBLP:conf/edbt/LiuCZZZQ17} to calculate the square, $sq_x=E_{pk}((l_w.x-l_t.x)^2)$, and $sq_y=E_{pk}((l_w.y-l_t.y)^2)$ \\
	$P_a$ calculates $sq_d=E_{pk}((l_w.x-l_t.x)^2) * E_{pk}((l_w.y-l_t.y)^2) = sq_x*sq_y = E_{pk}((l_w.x-l_t.x)^2 + (l_w.y-l_t.y)^2)$ \\
	\Return{$sq_d$}
	\caption{\texttt{SecureDistanceCalculation}}
	\label{algo:secureDistanceCalc}
\end{algorithm}

Step 4. After obtaining the true distance between workers and tasks, $P_b$ obtains an exact matching w.r.t. to the reachability graph constructed from the true distance. It returns the matching result to the server. Note that only the updated matching result is returned, the actual distance of the tasks are stored in $P_b$ only, and the true locations are always strictly protected and encrypted.

\section{Experiments}
\noindent\textbf{Real-world dataset}. For the real-world dataset, we use the taxi dataset from  The dataset contains detailed trajectory samples from the taxi cars running in Xi'an city, Shanxi province of China. On each day, the dataset contains around 30 million samples, with each data sample recording the GIS coordinates (latitude and longitude) of the taxi, the taxi ID, the passenger ID, and the timestamp. 

The sampled worker and task locations are in GIS coordinates (latitude and longitude), and we convert it to Cartesian coordinates in meters (in X and Y) using a common technique -- Equi-rectangular projection \cite{snyder1997flattening}, and shift the lower-left boundary point to (0,0). Our points are within the range $[0, 8004] \times [0,8259]$. 

\subsection{Additional results}
\label{subsec:appendix_experiments}

\begin{figure}[h!]\centering \figureTopMargin\vspace{0ex}
	\subfigure[][{\scriptsize Utility}]{
		\scalebox{0.20}[0.20]{\includegraphics{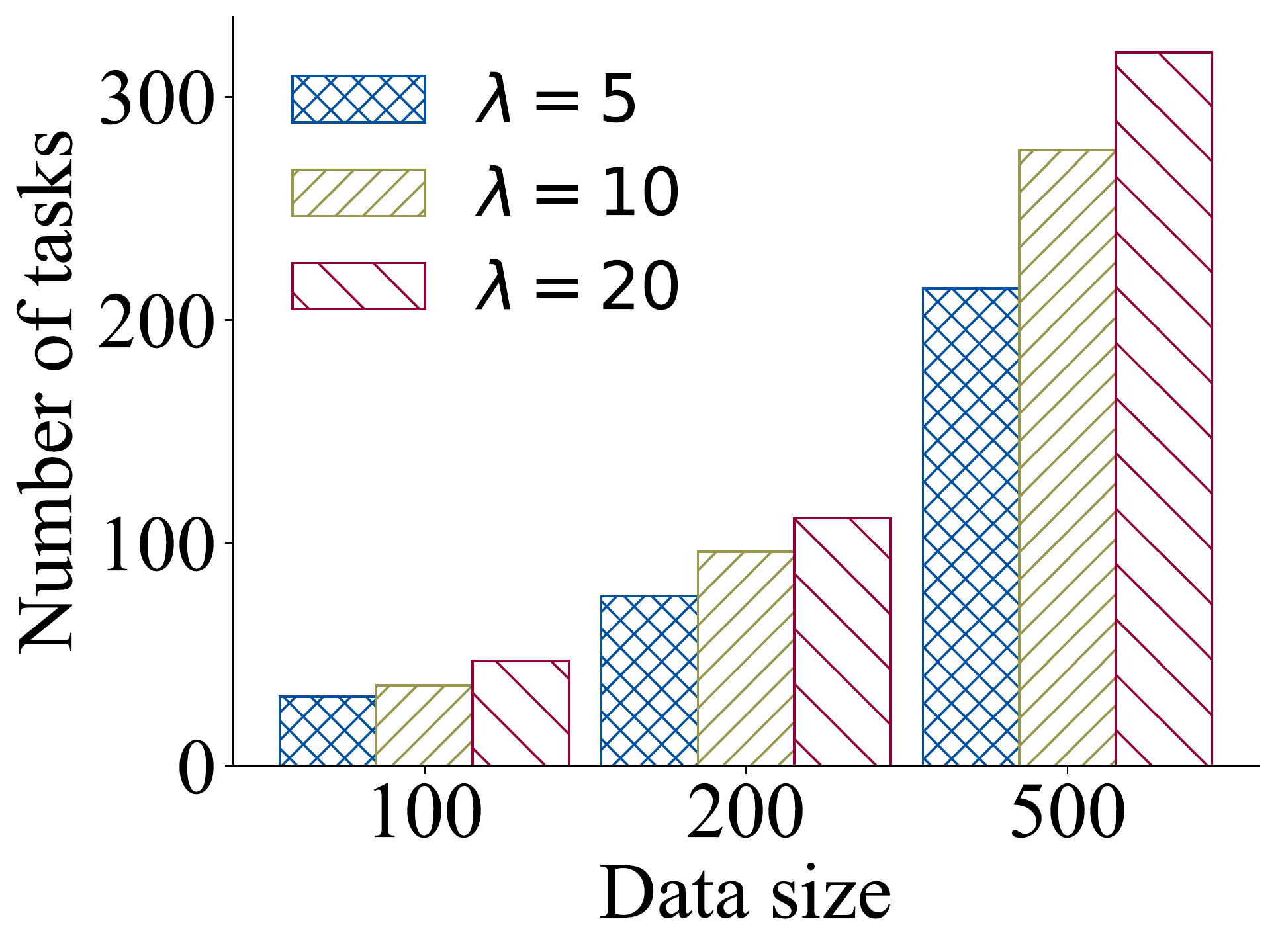}}
		\label{subfig:utility_lambda}}
	\subfigure[][{\scriptsize  Running time}]{
		\scalebox{0.20}[0.20]{\includegraphics{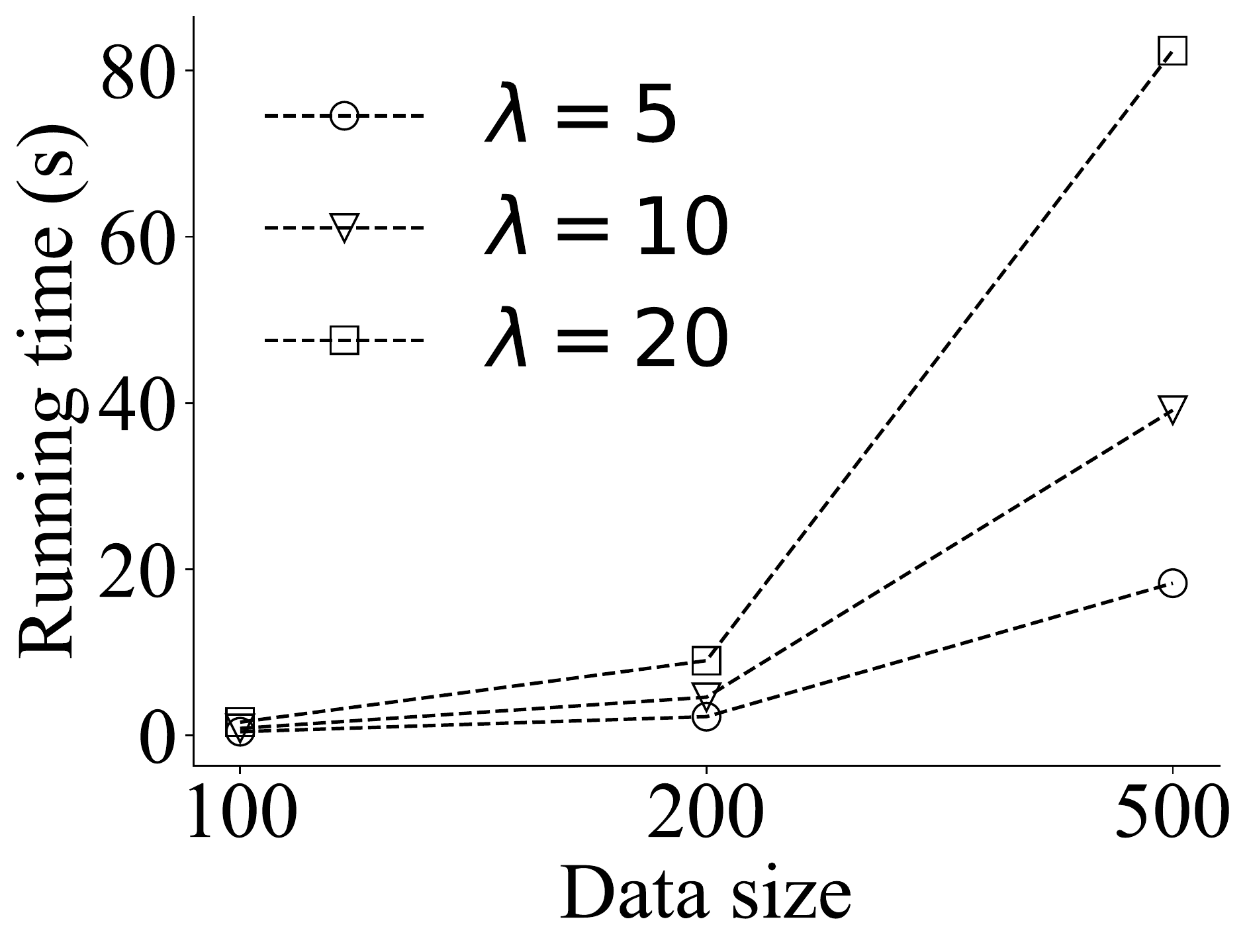}}
		\label{subfig:time_lambda}}\figureCaptionMargin
    \caption{\small Effect of parameter $\lambda$ for $k$-Switch method on the number of tasks assigned vs. the running time.}
    \figureBelowMargin\vspace{0ex}
	\label{fig:effect_lambda}
\end{figure}

\fakeparagraph{Varying parameter $\lambda$} 
Parameter $\lambda$ controls how many iterations $k$-Switch runs. For each iteration, $k$-Switch improves the number of assigned tasks by running $k$-HE protocol in small groups of size $k$. Thus $\lambda$ controls a tradeoff between the utility of our method vs. the running time. This is verified by the result shown in \figref{fig:effect_lambda}. When we reduce $\lambda$ from the system default ($\lambda=20$) to smaller values (\eg $\lambda=5, 10$), a moderate decrease of number of assigned tasks is observed, as shown in \figref{subfig:utility_lambda}. It is observed consistently across different input sizes. In the meantime, the running time of $k$-Switch is considerably improved, as shown in \figref{subfig:time_lambda}. The running time on data size 500 ($w, t=500$) is decreased from 82.4 seconds to 18.3 seconds when $\lambda$ is decreased from the system default ($\lambda=20$) to 5. 

\subsection{Discussions}

As the experiments show, $k$-Switch delivers strong performance in terms of significantly increasing the number of successfully assigned tasks of the baseline matching obtained by various oblivious baselines. Verifying the motivation of the research, $k$-Switch also outperforms the existing online method SCGuard by large margins. Despite more significant utility gain when we increase the value of $k$, we argue that setting $k=2$ is a cost-effective choice, as it balances the tradeoff between system overhead and utility gain. 

In addition, the secure computation we utilize within the small group is to prevent any location leak inside the group. Our experimental results are promising because a small value of $k$ (such as $k=2$) is sufficient for $k$-Switch to obtain significant utility gain. For real-world adoptions, even if there is malicious attack within the group such that the true locations are disclosed, the potential damage is manageable, given the fact that $k$ is small.

\end{document}